\documentclass[acmsmall,screen]{acmart}
\AtBeginDocument{%
  }

\setcopyright{acmlicensed}
\copyrightyear{2018}
\acmYear{2018}
\acmDOI{XXXXXXX.XXXXXXX}
\acmConference[Conference acronym 'XX]{Make sure to enter the correct
  conference title from your rights confirmation email}{June 03--05,
  2018}{Woodstock, NY}
\acmISBN{978-1-4503-XXXX-X/2018/06}

\usepackage{booktabs}   
\usepackage{graphicx}
\usepackage{subcaption} 
\usepackage{caption} 

\usepackage{hyperref}
\usepackage{bussproofs}
\usepackage{float}
\usepackage{tikz}
\usepackage{xspace}
\usepackage{xcolor}
\usepackage{pgfplots}
\pgfplotsset{compat=1.18}
\usepackage{geometry}

\usetikzlibrary{automata, positioning, arrows,calc, positioning, shapes, arrows.meta}
\usepackage{amsthm}
\usepackage[longend,ruled,linesnumbered]{algorithm2e}
\newtheorem{definition}{Definition}
\newtheorem{problem}{Problem}
\newtheorem{corollary}{Corollary}

\usepackage{amsmath}

\usepackage{amssymb}
\usepackage{subcaption}

\newcommand{\zak}[1]{{\color{red}Zak: #1}}

\newcommand{\kar}[2]{\ensuremath{KA\vDash~#1= #2}}
\newcommand{\karle}[2]{\ensuremath{KA\vDash~#1\le #2}}
\newcommand{\kae}[2]{\ensuremath{KA+C\vDash~#1= #2}}
\newcommand{\ka}[4]{\ensuremath{#1+#2\vDash~#3= #4}}
\newcommand{\kale}[4]{\ensuremath{#1+#2\vDash~#3\le #4}}
\newcommand{\kc}{\ensuremath{ KA+C }}
\newcommand{\kcs}{\ensuremath{ KA^*+C }}

\newcommand{\lc}{\ensuremath{L_C}}

\newcommand{\defeq}{\triangleq}
\newcommand{\Rat}[1]{\textsf{Rat}(#1)}

\begin{document}

\title[Kleene Algebra with Transitive Commutativity Conditions]{Kleene Algebra with Transitive Commutativity Conditions}         


\author{Han Xu}
\email{hx3501@princeton.edu}          
\orcid{0000-0002-2548-6866}

\affiliation{
  \institution{Princeton University}            
  \country{United States}                    
}
\author{Chenyu Zhou}
\email{czhou691@usc.edu}          
\orcid{0009-0006-8493-6886}
\affiliation{
  \institution{University of Southern California}            
  \country{United States}                    
}

\author{David Walker}
\email{dpw@princeton.edu}          
\orcid{0000-0003-3681-149X}

\affiliation{
  \institution{Princeton University}            
  \country{United States}                    
}

\author{Zachary Kincaid}
\email{zkincaid@cs.princeton.edu}          
\orcid{0000-0002-7294-9165}

\affiliation{
  \institution{Princeton University}            
  \country{United States}                    
}

\begin{abstract}

Kleene algebra ($KA$) provides a foundational algebraic framework for reasoning about program structure and control flow. To capture equivalences arising from reordering or independence of actions, \citet{Kozen96KAT} purposed that $KA$ can be extended with \emph{commutativity conditions}, that is, equations of the form
\(
\{\,ab=ba \mid (a,b)\in C\,\},
\)
where $C$ is a binary relation on constant symbols. This paper studies the following question: for which relations $C$ is the equational theory of $KA+C$ decidable?

Early related work~\cite{regtrans82,Ibarra78} showed that regular languages modulo commutativity conditions $C$ are decidable if and only if $C$ is transitive. For Kleene algebra $KA$ and commutativity conditions $C$, however, the situation is substantially more difficult. Only very recently, \citet{kuzuetsov23kac} showed that the equational theory of Kleene algebra $KA+C$ is undecidable under certain specific commutativity conditions, settling the first nontrivial cases more than 25 years after the corresponding problem for $KA^*+C$ was resolved by \citet{Kozen96KAT}. Nevertheless, the decidability problem of $KA+C$ remained open.

In this work, we resolve this question completely by showing that the equational theory of $KA+C$ is decidable if and only if $C$ is transitive. Moreover, we strengthen the result in both directions. On the negative side, we show that when $C$ is not transitive, the universality problem for $KA+C$ is already undecidable. On the positive side, we show that for transitive $C$, the equational theories of $KA^*+C$ and $KA+C$ coincide. \end{abstract}



\maketitle
\section{Introduction}

Kleene algebra ($KA$) provides an algebraic foundation for reasoning about the control structure of programs. Its operators—addition, multiplication, and Kleene star—correspond naturally to nondeterministic choice, sequential composition, and iteration. This correspondence makes $KA$ a powerful framework for expressing and verifying program equivalences through algebraic manipulation, and it has found applications in program verification, compiler optimization, network analysis and the study of regular languages~\cite{kozen1994ka,conway1971regular,kozen1997kat,netkat}.

However, the axioms of standard Kleene algebra, as well as those of its $*$-continuous variant $KA^*$, capture only \emph{path equivalence}, that is, equivalence of program traces viewed as regular languages. When reasoning about programs with partially independent or commuting actions—such as concurrent statements, database transactions, or reordering optimizations—this notion of equivalence is often too restrictive. To model such behaviors, one must extend the algebra with additional axioms.

One such extension is \emph{Kleene algebra with commutativity conditions} (\kc)~\cite{Kozen96KAT}, which augments the algebra with a specification of which pairs of atomic actions may commute. Formally, a commutativity condition $C$ asserts that for certain atomic actions $a$ and $b$, the equation $ab=ba$ holds, expressing that these two actions may be executed in either order without changing the overall program behavior. This extension allows the algebra to reason about partially independent computations and to capture reordering transformations that preserve program equivalence. The decidability of \kcs\ was fully characterized after \citet{kozen02comp} observed that \kcs\ has the same equational theory as regular languages with commutativity conditions, which are known to be decidable if and only if $C$ is transitive~\cite{Ibarra78,regtrans82}.

For \kc, however, the situation is substantially more difficult. Only very recently did \citet{kuzuetsov23kac} and \citet{azevedo25kac} independently show that the full theory of \kc\ is \emph{undecidable} when certain partial commutativity relations are allowed. Still, this leaves the central question unresolved: \emph{which fragments of Kleene algebra with commutativity conditions remain decidable?}

Two extreme cases are already well understood. When all actions commute, \kc\ becomes \emph{commutative Kleene algebra} ($CKA$), whose equational theory is decidable and coincides both with the theory of semilinear sets~\cite{pilling1970algebra} and with that of \(CKA^*\) (Theorem~\ref{thm:cka-completeness}). At the other extreme, when no commutativity constraints are present, \kc\ is the ordinary $KA$, whose equational theory is decidable and coincides exactly with both regular-language equivalence and the equational theory of $KA^*$~\cite{kozen1994ka}. These two endpoints led us to conjecture that $KA+C$ is decidable if and only if $C$ is transitive.

In this paper, we strengthen this conjecture in two directions. On the positive side, when $C$ is transitive, we prove that the equational theories of \kc\ and \kcs\ coincide, thereby generalizing the classical results for $KA$ and $CKA$. As a consequence, decidability of \kc\ follows from that of \kcs, namely, from the decidability of regular languages under transitive commutativity conditions.

On the negative side, we prove that universality—that is, whether an expression $e$ is equivalent to $\Sigma^*$—for $KA+C$ is already undecidable in the \emph{minimal} non-transitive setting, namely when the commutativity conditions satisfy $(a,b)\in C$ and $(b,c)\in C$, but $(a,c)\notin C$. This strictly improves previous constructions, which required a non-transitive commutativity relation over at least a four-letter alphabet, with
\[
(a,c),(a,d),(b,c),(b,d)\in C
\qquad\text{but}\qquad
(a,b),(c,d)\notin C,
\]
and established only undecidability of equivalence, that is, whether two expressions $e_1$ and $e_2$ are equal~\cite{azevedo25kac,kuzuetsov23kac}.

Taken together, our results completely settle the decidability of $KA+C$ by showing that \emph{transitivity is exactly the decidability boundary} for Kleene algebras with commutativity conditions. We strengthen both sides of this characterization: on the positive side, decidability follows from the coincidence of the equational theories of $KA+C$ and $KA^*+C$ when $C$ is transitive; on the negative side, undecidability is strengthened from equivalence to universality when $C$ is not transitive.
\section{Preliminaries}\label{sec:background}

In this section, we recall the basic definitions of Kleene Algebra and introduce the formal framework for reasoning about commutativity conditions.  
We then define the decision problem studied in this paper, namely, the equivalence of Kleene Algebra expressions under an equivalence relation generated by a given set of commuting pairs.  

\subsection{Kleene Algebra}
We begin by recalling the definition of a semiring.

\begin{definition}[Semiring and Idempotent Semiring]
A \emph{semiring} is a structure
\[
\mathcal{S} = (K, +, \cdot, 0, 1)
\]
such that:
\begin{enumerate}
  \item $(K,+,0)$ is a commutative monoid, i.e., for all $a,b,c\in K$,
  \[
    a+b=b+a,\qquad (a+b)+c=a+(b+c),\qquad 0+a=a;
  \]
  \item $(K,\cdot,1)$ is a monoid:
  \[
    (a\cdot b)\cdot c=a\cdot(b\cdot c),\qquad
    1\cdot a=a\cdot1=a;
  \]
  \item Multiplication distributes over addition, and $0$ is absorbing:
  \[
    a\cdot(b+c)=a\cdot b+a\cdot c,\qquad
    (a+b)\cdot c=a\cdot c+b\cdot c,\qquad
    a\cdot0=0\cdot a=0.
  \]
\end{enumerate}

The semiring is called \emph{idempotent} if addition is idempotent:
\[
a+a=a \qquad\text{for all } a\in K.
\]
In this case, $K$ carries a natural partial order defined by
\[
a \le b \;\Longleftrightarrow\; a+b=b.
\]
\end{definition}

\begin{definition}[Kleene Algebra]
A \emph{Kleene algebra} (KA) is an idempotent semiring
\[
(K,+,\cdot,0,1)
\]
equipped with a unary operation $(\,\cdot\,)^* : K \to K$ satisfying the
\emph{star axioms}:
\[
1 + a a^* \le a^*, \qquad
1 + a^* a \le a^*,
\]
\[
\text{if } a b \le b \text{ then } a^* b \le b, \qquad
\text{if } b a \le b \text{ then } b a^* \le b,
\]
for all $a,b\in K$.

A Kleene algebra is called \emph{\(*\)-continuous} if for all $a,b,c\in K$,
\[
a\, b^*\, c
\;=\;
\sum_{n\ge0} a\, b^n\, c,
\]
where the infinite sum denotes the supremum with respect to the natural order
$\le$.
\end{definition}

The standard example of a (\(*\)-continuous) Kleene algebra is the algebra of regular languages over a finite alphabet.  This example is generalized by algebras of rational subsets of a monoid (defined below).  Rational subsets of a finitely-generated free monoid are precisely regular languages; rational subsets of (partially) commutative monoids serve as models of Kleene algebra with commutativity conditions.

\begin{definition}[Kleene Algebra \Rat{M}]
Let \((M,\cdot_M,1_M)\) be a monoid.
The powerset \(2^{M}\) forms a \(*\)-continuous Kleene algebra under the
operations
\[
0 \;\triangleq\; \emptyset,
\qquad
1 \;\triangleq\; \{1_M\},
\qquad
X + Y \;\triangleq\; X \cup Y,
\]
\[
X \cdot Y \;\triangleq\; \{\,x \cdot_M y \mid x \in X,\; y \in Y\,\},
\qquad
X^{*} \;\triangleq\; \bigcup_{n \ge 0} X^{n},
\]
where \(X^{0} = \{1_M\}\) and \(X^{n+1} = X \cdot X^{n}\).

The \(*\)-continuous Kleene algebra \Rat{M} is defined as the smallest
subalgebra of \(2^{M}\) containing all singletons \(\{\{m\} \mid m \in M\}\).
It is called the Kleene algebra of \emph{rational subsets} of \(M\).
\end{definition}

\subsection{Commutativity Conditions}
Next we proceed to the definition of commutativity conditions.

\begin{definition}[Commutativity Conditions]

    Let \(\Sigma\) be a finite set of constant symbols.  
A \emph{Commutativity Condition} on \(\Sigma\) is a reflexive and symmetric binary relation
\[
C \;\subseteq\; \Sigma \times \Sigma
\]
such that for \((a,b)\in C\) we interpret \(a\) and \(b\) as \emph{commuting}, i.e.,
\[
a \cdot b = b \cdot a.
\]
In the following, we omit elements of a commutativity condition whose existence is implied by reflexivity and symmetry.

A commutativity condition \(C\) on \(\Sigma\) is said to be \emph{transitive} if for all \(a,b,c \in \Sigma\),
\[
(a, b)\in C \;\wedge\; (b, c)\in C \;\Rightarrow\; (a, c)\in C .
\]
Equivalently, if \(a\) commutes with both \(b\) and \(c\), then \(b\) and \(c\) must also commute.

\end{definition}

\begin{definition}[Word Equivalence under Commutativity Conditions]
Let \(C \subseteq \Sigma \times \Sigma\) be a commutativity condition.
We use $\equiv_C$ to denote the smallest congruence relation on $\Sigma^*$ that contains $C$.  That is, for words $w_1,w_2 \in \Sigma^*$, 
we have $w_1 \equiv_C w_2$ iff \(w_1\) can be transformed into \(w_2\)
by a finite sequence of swaps of adjacent letters \(ab \mapsto ba\) with
\((a,b) \in C\).
\end{definition}

When \(C\) is transitive, it induces a partition of \(\Sigma\) into classes of mutually commuting symbols.
However, when transitivity fails, the commutativity conditions can no longer be represented as such a partition.  
Lack of transitivity leads to \emph{undecidability} of equivalence in the corresponding *-continuous Kleene Algebra  as we will show later in Section \ref{sec:undec}.

With commutativity conditions in place, we can now define the  theories
of \kc\ and \kcs.

\begin{definition}[Interpretations and Theories]
Fix an alphabet \(\Sigma\) of constant symbols, and let \(T_\Sigma\) denote the
set of all regular expressions over \(\Sigma\).

An \emph{interpretation} over a Kleene algebra \(\mathcal{K}\) is a function
\[
I \colon \Sigma \to \mathcal{K},
\]
which extends to a function 
\[
I \colon T_\Sigma \to \mathcal{K},
\]
in a homomorphic way.

For expressions \(e_1,e_2 \in T_\Sigma\), we write
\[
\mathcal{K}, I \vDash e_1 = e_2
\]
if \(I(e_1)=I(e_2)\).

Let $C \subseteq \Sigma \times \Sigma$ be a commutativity condition.  We write
\[
\ka{KA}{C}{e_1}{e_2}
\]
(respectively,
\(\ka{KA^{*}}{C}{e_1}{e_2}\))
if for every Kleene algebra \(\mathcal{K}\)
(respectively, every \(*\)-continuous Kleene algebra \(\mathcal{K}\))
and every interpretation $I$ over \(\mathcal{K}\),
if $\mathcal{K},I \vDash ab = ba$ for all $(a,b) \in C$, then $\mathcal{K},I \vDash e_1 = e_2$.

We use $L_C(\cdot)$ to denote the interpretation over $\Rat{\Sigma^*/\equiv_C}$ that maps each \(a \in \Sigma\) to the singleton \(\{[a]\}\), where \([a]\) denotes the equivalence class of \(a\) in the quotient
monoid \(\Sigma^{*}/\equiv_C\).  Using $1$ to denote the identity relation on $\Sigma$, observe that $L_1(e)$ is the usual interpretation of $e$ as a regular language. In the subsequent sections, we use $L(e)$ to denote the interpretation $L_1(e)$.
\end{definition}

The language \(L_C(\cdot)\) will appear frequently in what follows. Intuitively, \(L_C(e)\) is the quotient of the language \(L(e)\) of an expression \(e\) under the equivalence relation \(\equiv_C\). Thus, when \(L_C(e_1)=L_C(e_2)\), it means that for every word \(w_1\in L(e_1)\), there exists a word \(w_2\in L(e_2)\) such that \(w_1\equiv_C w_2\), and conversely, for every word \(w_2\in L(e_2)\), there exists a word \(w_1\in L(e_1)\) such that \(w_1\equiv_C w_2\). In other words, the two languages consist of the same set of words up to \(C\)-equivalence.

\subsection{Parikh Images, Semi-linear Sets, and Commutative Kleene Algebra}

It is a classical result that the Parikh image of any regular language is a
semi-linear set, and that such images can be computed effectively
\cite{parikh1966context,badban2010semilinear,to2010parikh}.

\begin{definition}[Semi-linear Set]
A set $S \subseteq \mathbb{N}^k$ is \emph{semi-linear} if it is a finite union of
linear sets.  That is, $S$ is semi-linear if there exist vectors
$\mathbf{b}_i, \mathbf{p}_{ij} \in \mathbb{N}^k$ such that
\[
S
  = \bigcup_{i=1}^n
      \left\{
        \mathbf{b}_i + \sum_{j=1}^{m_i} n_j \mathbf{p}_{ij}
        \;\middle|\;
        n_j \in \mathbb{N}
      \right\}.
\]
\end{definition}

\paragraph{Parikh Image.}
Let $\Sigma = \{a_1,\dots,a_k\}$ be a finite alphabet.
The \emph{Parikh vector} of a word $w\in\Sigma^*$ is the vector
$\Psi(w) \in \mathbb{N}^k$ where $\Psi(w)_i$ is the number of occurrences of
$a_i$ in $w$.  
For a language $L \subseteq \Sigma^*$, its Parikh image is
\[
P(L) = \{\, \Psi(w) \mid w \in L \,\} \subseteq \mathbb{N}^k.
\]

By Parikh's theorem~\cite{parikh1966context}, for every regular language $L$
the set $P(L)$ is semi-linear.

\paragraph{Commutative Kleene Algebra.}
A \emph{commutative} Kleene algebra is a Kleene algebra satisfying the axiom
\[
\forall p,q,~ p \cdot q = q \cdot p.
\]
We use CKA to denote the first-order theory of commutative Kleene algebras.
Prior work has~\cite{pilling1970algebra,Brunet19note} proved that, under this axiom, the equational theory
of Kleene algebra collapses precisely to equality of Parikh images:

\begin{theorem}\label{thm:cka-completeness}[Parikh Image over Commutative Kleene Algebra \cite[Lemma 4.11]{pilling1970algebra}]
For all expressions $e_1,e_2$,
\[
CKA~\vDash{e_1}={e_2}
\quad\Longleftrightarrow\quad
P(L(e_1)) = P(L(e_2)).
\]
\end{theorem}

Thus, checking whether an equation is valid for all commutative Kleene algebras can be reduced to checking equality of semi-linear sets (i.e., equivalence of Presburger formulas).

\begin{theorem}[Coincidence of $CKA$ and $CKA^{*}$]\label{thm:cka-coin}
For all expressions $e_1,e_2\in T_{\Sigma}$,
\[
CKA \vDash e_1 = e_2
\quad\Longleftrightarrow\quad
CKA^{*} \vDash e_1 = e_2 .
\]
\end{theorem}

\begin{proof}
It is easy to verify that
\[
P(L(e_1)) = P(L(e_2)) \Longleftrightarrow \Rat{\mathbb{N}^{\Sigma}},I \vDash e_1=e_2
\]
where $I$ is the interpretation wherein for each $a \in \Sigma$,
$I(a)_a = 1$ and $I(a)_b = 0$ for $b \neq a$.
Since $\Rat{\mathbb{N}^{\Sigma}}$ is a *-continuous Kleene algebra, the coincidence is then a trivial consequence of Theorem \ref{thm:cka-completeness}. See full proof in the Appendix.
\end{proof}

\subsection{Language Equivalence under Commutativity Conditions}

Having defined the equational theories of \(KA+C\) and \(KA^{*}+C\), we now turn to a language-theoretic formulation of equivalence in terms of \(L(\cdot)\) and \(L_C(\cdot)\) in the following sections. This perspective allows us to state and prove our main results in a more transparent and constructive way.

The main justification for this shift is a classical theorem of \citet{kozen02comp}, which shows that for \( * \)-continuous Kleene algebras, equational provability coincides with language equivalence. As an immediate consequence, it follows that \(KA^{*}+C\) is decidable if and only if \(C\) is transitive.

\begin{theorem}\label{thm:equiv}[Monoid Equations \cite[Lemma 4.1]{kozen02comp}]
Let \(\Sigma\) be a finite alphabet and let \(E\) be a finite set of equations
between words in \(\Sigma^{*}\).
Then, for all expressions \(e_1,e_2\),
\[
{\mathsf{Rat}(\Sigma^{*}/E)},{L_E}\vDash{e_1}={e_2}
\quad\Longleftrightarrow\quad
\ka{KA^{*}}{E}{e_1}{e_2}.
\]
where $L_E$ is the interpretation over $\Rat{\Sigma^*/E}$ that maps each \(a \in \Sigma\) to the equivalence class of \(a\) in the quotient
monoid \(\Sigma^{*}/E\).
\end{theorem}

\begin{corollary}[Decidability of \kcs]
The problem
\[
\ka{KA^*}{C}{e_1}{e_2}
\]
is decidable if and only if $C$ is transitive.
\end{corollary}

\begin{proof}
\citet{Ibarra78} and \citet{regtrans82} showed that equivalence of regular languages under commutativity conditions is decidable if and only if $C$ is transitive. The result therefore follows immediately from Theorem~\ref{thm:equiv}.
\end{proof}

\paragraph{Decision problem.}
The remaining question is therefore the decidability of $KA+C$. The central decision problem studied in this paper is the following.

\begin{problem}[Equivalence under Commutativity Conditions]
Given a finite alphabet \(\Sigma\), a commutativity condition \(C \subseteq \Sigma \times \Sigma\), and two regular expressions \(e_1,e_2\) over \(T_\Sigma\), decide whether
\[
\kae{e_1}{e_2}.
\]
\end{problem}
\section{Coincidence}\label{sec:coin}
In this section we prove that when the commutativity condition \(C\) is
transitive, the equational theories of $KA + C$ and $KA^* + C$ coincide.  Formally, for all expressions
\(p,q\),
\[
\ka{KA}{C}{p}{q}
\;\Longleftrightarrow\;
\ka{KA^*}{C}{p}{q}.
\]

By Theorem~\ref{thm:equiv}, we already know that
\[
\lc(p)= \lc(q)
\;\Longleftrightarrow\;
\ka{KA^*}{C}{p}{q}.
\]
Thus it suffices to show
\[
\lc(p)= \lc(q)
\;\Longleftrightarrow\;
\ka{KA}{C}{p}{q}.
\]

Our argument is partially inspired by Kozen’s proof of the coincidence of the equational theories of KAT and $*$-continuous KAT~\cite{kozen97kat}. In particular, it suggests that, rather than comparing expressions directly, it is more convenient to compare their normalized forms.

The first step is therefore to show that every Kleene algebra expression $p$ is provably equivalent to a canonical \emph{normal form} $\widetilde{p}$:
\[
\ka{KA}{C}{p}{\widetilde{p}}.
\]

Next, we show that for \textit{normal} expressions \(\widetilde{p}\) and \(\widetilde{q}\), 
$C$-equivalence of languages collapses to ordinary language equivalence:
\[
\lc(\widetilde{p})= \lc(\widetilde{q})
\;\Longleftrightarrow\;
L(\widetilde{p}) = L(\widetilde{q}).
\]

By the completeness of Kleene algebra~\cite{kozen1994ka}, we then obtain
\[
\ka{KA}{C}{\widetilde{p}}{\widetilde{q}}
\;\Longleftrightarrow\;
L(\widetilde{p}) = L(\widetilde{q}).
\]

Combining the equivalences above, we conclude the desired coincidence:
\[
\ka{KA}{C}{p}{q}
\;\Longleftrightarrow\;
\lc(p)= \lc(q).
\]

\subsection{Preparation}

Our commutativity assumption is alphabet–based: we require $ab=ba$ only for
letters $(a,b)\in C$, whereas commutative Kleene algebra assumes $pq=qp$ for
\emph{all} terms $p,q$.  To appeal to the completeness theorem for
commutative KA, we first show that alphabet–level commutativity already
forces full commutativity inside each class.

When $C$ is transitive, it forms an equivalence relation over $\Sigma$. Let $\Sigma/C=\{\Sigma_1,\ldots,\Sigma_m\}$ be the partition of $\Sigma$
induced by~$C$.
We prove that
\[
KA + \{ab=ba \mid a,b\in\Sigma_i\}\;\vDash\; pq=qp
\qquad\text{for all } p,q\in T_{\Sigma_i},
\]
i.e., commutativity of generators implies commutativity of all expressions in
the class.

\begin{lemma}\label{thm:alphabet-commutative}
Let $\Sigma_i$ be one
equivalence class under~$C$.
For any expressions $p,q \in T_{\Sigma_i}$,
\[
  \ka{KA}{\{ab = ba : a,b \in \Sigma_i\}}{pq}{qp}.
\]
\end{lemma}

\begin{proof} See Appendix.
\end{proof}

\subsection{Factorization}

Under the commutativity condition \(C\), two letters may commute only when they
belong to the same equivalence class \(\Sigma_i\) of the alphabet partition
\(\{\Sigma_1,\ldots,\Sigma_m\}\); letters from different classes do not commute.
This observation naturally leads to a two-stage normalization strategy:
\begin{itemize}
    \item[(1)] \emph{Factorize} each expression with respect to the alphabet
          partition \(\{\Sigma_1,\ldots,\Sigma_m\}\), rewriting every word as a sequence of blocks drawn from
          individual classes \(\Sigma_i\);
    \item[(2)] Normalize each class \(\Sigma_i\) separately using the
          \emph{supporting expressions} introduced later.
\end{itemize}

In this subsection we focus on stage~(1).
Given an expression \(p\), we construct its factorized form \(\hat{p}\), which
explicitly separates all words according to the alphabet partition.
The next subsection explains how to convert \(\hat{p}\) into the final normal
form \(\tilde{p}\).

Here we follow the Kleene algebra convention that the constant \(1\)
corresponds to the empty string \(\epsilon\) in regular languages.
We say that an expression \(e\) is \emph{$\epsilon$-free } if
\(1 \notin L(e)\).  
With this convention in place, we now introduce our definition of
factorization.

\begin{definition}[Factorization under $C$]\label{def:c-fact}
For a word $w \in \Sigma^{*}$, a \emph{factorization under $C$} is a
decomposition
\[
  w = u_{1} u_{2} \cdots u_{n},
\]
where each segment \(u_i\) is called a \emph{block}.  The factorization
satisfies:
\begin{itemize}
  \item if \(w = 1\), then the factorization consists of the single block
        \(1 = 1\);
  \item if \(w \neq 1\), each block \(u_i\) is a nonempty word drawn entirely
        from a single equivalence class, i.e.\ \(u_i \in \Sigma_{k_i}^{+}\)
        for some \(k_i \in \{1,\ldots,m\}\);
  \item adjacent blocks come from different classes:
        \(k_i \neq k_{i+1}\) for all \(i\).
\end{itemize}
\end{definition}

\begin{lemma}\label{thm:uniq-fac}
Every word \(w \in \Sigma^{*}\) admits exactly one factorization under \(C\).
\end{lemma}

\begin{proof}
 See Appendix.
\end{proof}

Our goal is to transform any expression \(p\) into an expression \(\hat{p}\) such that, by construction, we can directly read off \(\hat{p}\) to know how each word \(w\in L(\hat{p})\) is factorized. This makes it possible to reason about words block by block, aligned with the alphabet partition induced by the commutativity condition \(C\).

Before giving the factorization procedure, we recall two classical results that
we will make essential use of.  Both are due to Kozen~\cite{kozen1994ka}.

\begin{theorem}\label{thm:ka-matrix}[Matrix Kleene algebra\cite[Theorem 11]{kozen1994ka}]
Let $K$ be a Kleene algebra, and
let $M(n,K)$ be the set of all $n \times n$ matrices with entries in $K$.  Equipped with matrix addition, matrix multiplication, Kleene
star, the zero matrix $Z_n$, and the identity matrix $I_n$, the structure
\[
  \bigl(M(n,K),\; +,\; \cdot,\; {}^{*},\; Z_n,\; I_n\bigr)
\]
is itself a Kleene algebra. With $+$ and $\cdot$ being the matrix addition and multiplication, and the Kleene star of a matrix is defined inductively on its dimension.
For $n=1$, if $A=(a)$ is a $1\times 1$ matrix, define
\[
A^{*} = (a^{*}).
\]

For $n>1$, write $A$ in block form
\[
A =
\begin{pmatrix}
B & C \\
D & E
\end{pmatrix},
\]
where $B$ is a $k\times k$ matrix and $E$ is an $(n-k)\times(n-k)$ matrix
for some $1\le k<n$. Set
\[
A^{*}
=
\begin{pmatrix}
F & G \\
H & J
\end{pmatrix},
\]
where
\[
\begin{aligned}
F &= (B + C E^{*} D)^{*}, \qquad&
G &= F C E^{*}, \\
H &= E^{*} D F, &
J &= E^{*} + E^{*} D F C E^{*}.
\end{aligned}
\]
This matrix $A^*$ is unique regardless of the choice of $k$.
\end{theorem}

\begin{lemma}\label{thm:matrix-rep}[Matrix Representation of Expressions\cite[Lemma 15]{kozen1994ka}]
For every regular expression $p\in T_\Sigma$, there exist a natural number \(n\),
vectors \(u,v \in \{0,1\}^n\), and an \(n \times n\) 0-1 matrix \(A\) over \(T_\Sigma\) such
that
\[
\ka{KA}{C}{p}{u^{T} A^{*} v}.
\]
Here \(u^{T} A^{*} v\) is syntactic sugar for
\(
\sum_{i,j} u_i\, (A^{*})_{ij}\, v_j
\),
and the matrix \(A\) has the form
\[
A \;=\; \sum_{a \in \Sigma} a \cdot A_a,
\]
where each \(A_a\) is a \(0\)--\(1\) matrix encoding the transitions labeled by
the symbol \(a\).
\end{lemma}

Before proceeding to the proofs, we fix some notation.
For vectors $u,v \in \{0,1\}^n$ and an $n\times n$ matrix $A$ over $T_\Sigma$, we write\[
u^{T} A v \;\triangleq\; \sum\limits_{i,j} u_i\, (A)_{ij}\, v_j .
\]By the construction in Theorem~\ref{thm:ka-matrix}, if $A$, $A_1$, and $A_2$ are
matrices whose entries are drawn from $T_\Sigma$, then so are
$A^{*}$, $A_1 + A_2$, and $A_1 A_2$.
Since every $0$--$1$ matrix is trivially a matrix of  entries drawn from $T_\Sigma$, it
follows that all matrices arising in our constructions can be assumed to have entries drawn from $T_\Sigma$.

Finally, since every Kleene algebra forms a semiring, equality of matrices is
preserved under multiplication by vectors.
In particular, if $A_{ij}$ denotes the $(ij)$-entry of a matrix $A$, then
whenever \(\forall i,j,\,
KA \vDash (A_1)_{ij} = (A_2)_{ij},
\) we also have
\(
KA \vDash u^{T} A_1 v = u^{T} A_2 v,
\)
which follows by a direct unrolling of matrix multiplication.

We now define how to factorize a regular expression $p$ with respect to an alphabet
          partition \(\{\Sigma_1,\ldots,\Sigma_m\}\).

\begin{definition}[Factorization]
Let \(p\) be a regular expression with matrix representation
\(p = u^{T} A^* v\), where
\[
A = \sum_{a \in \Sigma} a \cdot A_a
\]
as given by Lemma~\ref{thm:matrix-rep}.
Without loss of generality, we may suppose that $A$ is of the  dimension $n$.
For each class \(\Sigma_i\) of the partition \(\Sigma/C\), define
\[
A_i \;=\; \sum_{a \in \Sigma_i} a \cdot A_a.
\]

We construct the block matrix \(\hat{A}\in M_{mn}(T_\Sigma)\) of dimension
\((m n)\times(m n)\) as follows:
\[
\hat{A} \;=\;
\begin{pmatrix}
0      & A_1^{+} & A_1^{+} & \cdots & A_1^{+} \\
A_2^{+}& 0       & A_2^{+} & \cdots & A_2^{+} \\
A_3^{+}& A_3^{+} & 0       & \cdots & A_3^{+} \\
\vdots & \vdots  & \vdots  & \ddots & \vdots  \\
A_m^{+}& A_m^{+} & A_m^{+} & \cdots & 0
\end{pmatrix}.
\]

Let \(v^{(m)}\) denote the vertical concatenation of \(m\) copies of \(v\):
\[
v^{(m)} =
\begin{pmatrix}
v \\ v \\ \vdots \\ v
\end{pmatrix},
\qquad
u^{(m)} \text{ defined analogously}.
\]

We then define the \emph{factorization} of \(p\) to be
\(
\hat{p}
\;=\;
\bigl(u^{(m)}\bigr)^{T}\, (\hat{A})^{\,*}\, v^{(m)}.
\)
\end{definition}
The following useful property is immediate and will be preserved throughout all
subsequent constructions.

\begin{lemma}[$\epsilon$-free $\hat{A}$]\label{thm:epsilon-free-hatA}
For the matrix $\hat{A}$ constructed above, every entry $\hat{A}_{i,j}$ is
\emph{$\epsilon$-free}, i.e.\ $1 \not\leq \hat{A}_{i,j}$ (or equivalently
$1 \notin L(\hat{A}_{i,j})$).
\end{lemma}

\begin{proof}
Each block
\(
A_i \;=\; \sum_{a\in \Sigma_i} a \cdot A_a
\)
is a finite sum of concrete letters (or~$0$).  
Hence no entry of $A_i$ accepts the empty word; that is, $1\notin L(e)$ for every
entry $e$ of $A_i$.

Since $A_i^{+} = A_i A_i^{*}$, any entry of $A_i^{+}$ is either $0$ or contains at least one
letter from $\Sigma_i$, and therefore remains $\epsilon$-free.
Because $\hat{A}$ is assembled entirely from these blocks $A_i^{+}$ and zeros,
every entry of $\hat{A}$ is $\epsilon$-free as well.
\end{proof}

Then we can show that our first factorization step yields a provably equivalent expression.
\begin{theorem}\label{thm:regular-factorization}
For every expression \(p\) and its factorized form \(\hat{p}\), we have
\(
L(p)=L(\hat{p}).
\)
\end{theorem}

\begin{proof}
Let
\(
p = u^{T} A^{*} v,
\) with
\(
A = \sum_{a \in \Sigma} a \cdot A_a,
\)
be the matrix representation of \(p\) constructed by Lemma \ref{thm:matrix-rep}, and let
\(
\hat{p}
= \bigl(u^{(m)}\bigr)^{T}\, (\hat{A})^{*}\, v^{(m)}
\)
be its factorized form. By standard unrolling of matrix multiplication and the construction of
\(\hat{A}\), one can proof both inclusions \(L(p)\subseteq L(\hat{p})\) and \(L(\hat{p})\subseteq L(p)\). The full argument is deferred to the Appendix.
\end{proof}

\begin{corollary}\label{thm:ka-factorization}
For every expression \(p\) and its factorized form \(\hat{p}\), we have
\[
\kae{p}{\hat{p}}.
\]
\end{corollary}

\begin{proof}
We have already shown that \(L(p) = L(\hat{p})\).  
By the completeness theorem for Kleene algebra~\cite{kozen1994ka}, language
equality implies equational provability in Kleene Algebra (even without commutativity conditions).  
Hence \(\kae{p}{\hat{p}}.\)
\end{proof}

\subsection{Supporting Expressions}
Since factorization has already restructured the expression by separating different sub-alphabets, we now turn to normalization \emph{within a single sub-alphabet}~$\Sigma_i$.
The goal of this stage is to decompose each expression over~$\Sigma_i$
into ``atomic'' components whose Parikh images are pairwise disjoint under
commutation.  Formally:

\begin{problem}
Given finitely many expressions $p_1,\ldots,p_n$ over the sub-alphabet
$\Sigma_i$, find expressions $q_1,\ldots,q_m$ over the sub-alphabet $\Sigma_i$ such that for every $p_j$ there
is an index set $I_j \subseteq \{1,\ldots,m\}$ satisfying
\[
\ka{KA}{C}{p_j}{\sum_{k\in I_j} q_k},
\quad\text{and}\quad
\lc(q_k)\; \cap\; \lc(q_{k'}) = \emptyset \;\text{ for } k\neq k'.
\]
\end{problem}
Once such atomic components $\{q_k\}$ are constructed inside each
sub-alphabet~$\Sigma_i$, the normal form $\tilde{p}$ of an expression $p$ is
obtained simply by replacing each subexpression $p_j$ with its canonical
disjoint decomposition $\sum_{k\in I_j} q_k$.

Equivalently, the second condition can be written as follows:
\begin{lemma}\label{thm:construction-disjointness}
If $q_k, q_{k'} \in T_{\Sigma_i}$, then
\(
L_C(q_k)\cap L_C(q_{k'})=\emptyset
\quad\Longleftrightarrow\quad
P(L(q_k))\cap P(L(q_{k'}))=\emptyset .
\)
\end{lemma}

\begin{proof}
Since both expressions are over the same sub-alphabet $\Sigma_i$, two words are
$C$-equivalent iff they have the same Parikh image.
Thus $L_C(q_k)\cap L_C(q_{k'})\neq\emptyset$ holds iff there
exist $w_k\in L(q_k)$ and $w_{k'}\in L(q_{k'})$ with $\Psi(w_k)=\Psi(w_{k'})$, which is
equivalent to $P(L(q_k))\cap P(L(q_{k'}))\neq\emptyset$.
\end{proof}

Thus our construction proceeds in two steps:  
first, we generate \emph{supporting sets} \(\{S_1,\dots,S_n\}\) and \emph{atomic sets} \(\{B_1,\dots,B_m\}\) at the level of Parikh
images (i.e.\ semi-linear sets); next, we apply an inverse Parikh-image
construction to obtain the corresponding expressions, which we call
\emph{atomic expressions} \(\{e_1,\dots,e_m\}\).  
To describe supporting sets, we recall the following standard result:

\begin{lemma}[Finite Partition Induced by a Finite Family of Sets]
Let $U$ be any set, and let $S_1,\dots,S_n\subseteq U$ be finitely many subsets.
Then there exist finitely many (unique up to permutation) non-empty subsets $B_1,\dots,B_m \subseteq U$ such that:
\begin{enumerate}
  \item[\textup{(1)}] (\emph{Disjointness})
    $B_i \cap B_j = \emptyset$ for all $i\neq j$.
  \item[\textup{(2)}] (\emph{Representation})
    Each $S_i$ is expressible as a union
    $S_i = \bigcup_{k\in I_i} B_k$ for a unique $I_i\subseteq\{1,\ldots,m\}$.
  \item[\textup{(3)}] (\emph{Atomic decomposition})
    Every $B_k$ is a Boolean minterm:
    \[
      B_k
      =
      \Bigl(\bigcap_{i\in P_k} S_i\Bigr)
      \cap
      \Bigl(\bigcap_{j\in N_k} (U\setminus S_j)\Bigr),
    \]
    for some $P_k,N_k\subseteq\{1,\ldots,n\}$ with $P_k\cap N_k=\emptyset$.
\end{enumerate}
Thus $\{B_1,\dots,B_m\}$ is a finite partition of $\bigcup_{i=1}^n S_i$ into
pairwise disjoint Boolean minterms.
\end{lemma}
\begin{proof}
This follows from the classical fact that any finitely-generated Boolean algebra has
finitely many atoms, corresponding exactly to the Boolean minterms generated by
\(\{S_1,\dots,S_n\}\).  
See, e.g., \cite{Sikorski1960-SIKBA-2,Halmos1966-HALLOB,Davey_Priestley_2002}.
\end{proof}

The family $\{S_i\}$ and $\{B_j\}$ is precisely the \emph{supporting sets} and \emph{atomic sets} we require.
Since intersection, complement, and set difference of semilinear sets are all
semilinear, the lemma guarantees that whenever $S_1,\ldots,S_n$ are semilinear,
the resulting blocks $\{B_i\}$ also form a finite family of semilinear sets.

Next we construct an \emph{inverse Parikh image} $P^{-1}$ that maps any
semilinear set back to a regular expression.

\begin{lemma}
For every semilinear set \(S\subseteq\mathbb{N}^k\), there exists a regular
expression \(e\) such that \(P(L(e))=S\).
\end{lemma}

\begin{proof}
See Appendix.
\end{proof}

This allows us to define the inverse Parikh map in a canonical way.

\begin{definition}[Inverse Parikh Image]
The inverse Parikh image \(P^{-1}\) maps any semilinear set \(S\) to the
lexicographically smallest regular expression \(e\) such that \(P(L(e))=S\).
\end{definition}

We now define the decomposition of expressions over a fixed sub-alphabet
\(\Sigma_i\).

\begin{definition}[Decomposition inside $\Sigma_i$]
Let $p_1,\ldots,p_n$ be expressions over the sub-alphabet $\Sigma_i$, and let
$B_1,\ldots,B_m$ be the finitely many atomic subsets (Boolean minterms)
generated by the Parikh images $P(L(p_1)),\ldots,P(L(p_n))$ with respect to any given superset $U$.
For each atom $B_k$, define its corresponding expression
\[
  q_k \;=\; P^{-1}(B_k).
\]

Since every Parikh image $P(L(p_i))$ admits a unique disjoint decomposition
\[
  P(L(p_i)) \;=\; \bigcup_{k\in I_i} B_k
  \qquad\text{for a uniquely determined } I_i \subseteq \{1,\ldots,m\},
\]
we define the rewritten expression
\[
  p_i' \;=\; \sum_{k\in I_i} q_k,
  \qquad\text{with } p_i' := 0 \text{ if } I_i = \emptyset.
\]

The family $\{q_1,\ldots,q_m\}$ is called the set of \emph{atomic expressions}
for the \emph{supporting expressions} $\{p_1,\ldots,p_n\}$ inside the sub-alphabet~$\Sigma_i$, and the
expressions $\{p_1',\ldots,p_n'\}$ form the corresponding \emph{term
rewriting} of $\{p_1,\ldots,p_n\}$.
\end{definition}

\begin{lemma}\label{thm:supp-equiv}
Let $\{q_1,\ldots,q_m\}$ and $\{p_1',\ldots,p_n'\}$be the atomic expression and rewriting terms for
$\{p_1,\ldots,p_n\}$ inside~$\Sigma_i$. Then:
\begin{itemize}
  \item[(1)] $\ka{KA}{C}{p_i}{p_i'}$ for all $i$;
  \item[(2)] $\lc(q_k)\,\cap\,\lc(q_{k'})=\emptyset$ whenever $k\ne k'$.
\end{itemize}
\end{lemma}

\begin{proof}
\textbf{(1)}  
By construction, $P(L(p_i'))=P(L(p_i))$.  
Thus, by Lemma~\ref{thm:alphabet-commutative} and the completeness of
commutative Kleene algebra (Theorem~\ref{thm:cka-completeness}),
$\ka{KA}{C}{p_i}{p_i'}$.

\smallskip
\textbf{(2)}  
Each $q_k$ satisfies $P(L(q_k))=B_k$, and the Boolean blocks $B_k$ are disjoint.
Hence $P(L(q_k))\cap P(L(q_{k'}))=\emptyset$ for $k\neq k'$
by Lemma \ref{thm:construction-disjointness}.  Thus $\lc(q_k)\,\cap\,\lc(q_{k'})=\emptyset$.
\end{proof}

Now we define supporting expressions and the normalization procedure for
a factorized expression~\(\hat{p}\).

\begin{definition}[Supporting expressions and normalization]
Let
\[
\hat{p}
\;=\;
\bigl(u^{(m)}\bigr)^{T}\, (\hat{A})^{*}\, v^{(m)},
\]
where
\[
\hat{A} \;=\;
\begin{pmatrix}
0        & A_1^{+} & A_1^{+} & \cdots & A_1^{+} \\
A_2^{+}  & 0       & A_2^{+} & \cdots & A_2^{+} \\
A_3^{+}  & A_3^{+} & 0       & \cdots & A_3^{+} \\
\vdots   & \vdots  & \vdots  & \ddots & \vdots  \\
A_m^{+}  & A_m^{+} & A_m^{+} & \cdots & 0
\end{pmatrix}
\]
is obtained from the factorization step, and each
\(A_i^{+} = (\sum_{a\in\Sigma_i} a\cdot A_a)^{+}\).

\paragraph{(1) Supporting expressions.}
For each subalphabet \(\Sigma_i\), define the supporting set
\[
\mathrm{Support}_i(\hat{p})
   \;=\;
   \{\, a_{pq} \mid a_{pq}\ \text{is an entry of } A_i^{+} \,\}.
\]
This supporting set collects expressions that must later be rewritten as sums of disjoint atomic expressions.

\paragraph{(2) Normalization inside \(\Sigma_i\).}
Let \(S_i \supseteq \mathrm{Support}_i(\hat{p})\) be any set containing the expressions over \(\Sigma_i\) that need to be rewritten later.
Applying the decomposition procedure to every entry \(a_{pq}\in A_i^{+}\)
with respect to \(S_i\) yields its term rewriting \(a_{pq}'\).
We write
\[
\tilde{A}_i \;=\; (\,a_{pq}'\,)_{1\le p,q\le n}
\]
for the matrix obtained by replacing each entry \(a_{pq}\) of \(A_i^{+}\) by its decomposition into a sum of atomic expressions with respect to $S_i$.

\paragraph{(3) Global normalization.}
Define the normalized block matrix
\[
\tilde{A}
\;=\;
\begin{pmatrix}
0           & \tilde{A}_1 & \tilde{A}_1 & \cdots      & \tilde{A}_1 \\
\tilde{A}_2 & 0           & \tilde{A}_2 & \cdots      & \tilde{A}_2 \\
\tilde{A}_3 & \tilde{A}_3 & 0           & \cdots      & \tilde{A}_3 \\
\vdots      & \vdots      & \vdots      & \ddots      & \vdots      \\
\tilde{A}_m & \tilde{A}_m & \tilde{A}_m & \cdots      & 0
\end{pmatrix}.
\]

The \emph{normal form} of \(\hat{p}\) with respect to the family of expression sets
\(\{S_1,\ldots,S_m\}\) is
\[
\tilde{p}
\;=\;
\bigl(u^{(m)}\bigr)^{T}\, (\tilde{A})^{*}\, v^{(m)}.
\]
Here, for each \(i\), the set \(S_i\) is used as the collection of expressions with respect to which the matrix \(A_i^{+}\) is rewritten.
\end{definition}
By construction, we immediately conclude that this rewriting preserves the \(\epsilon\)-free property and yields a provably equivalent expression.

\begin{lemma}[$\epsilon$-free $\tilde{A}$]\label{thm:epsilon-free-tildeA}
Every entry of the matrix $\tilde{A}$ constructed above is
\emph{$\epsilon$-free}, i.e., $ 1 \not\leq \tilde{A}_{i,j}$.
\end{lemma}

\begin{proof}
Each entry $e_{ij}'$ of $\tilde{A}$ is obtained from the corresponding
entry $e_{ij}$ of $\hat{A}$ by rewriting with expressions having exactly the
same Parikh image.  Hence
\(
P(L(e_{ij}')) = P(L(e_{ij})).
\)
By Lemma~\ref{thm:epsilon-free-hatA}, every $e_{ij}$ is $\epsilon$-free, so
\[
1 \notin L(e_{ij})
~\Longrightarrow~
(0,0,\ldots,0) \notin P(L(e_{ij}))
= P(L(e_{ij}'))
~\Longrightarrow~
1 \notin L(e_{ij}').
\]
Thus every entry of $\tilde{A}$ is $\epsilon$-free.
\end{proof}

\begin{theorem}
For any finite family of supporting sets \(\{S_1,\ldots,S_m\}\), the
normalized expression \(\tilde{p}\) satisfies
\(
\kae{\hat{p}}{\tilde{p}}.
\)
\end{theorem}

\begin{proof}
By Lemma~\ref{thm:supp-equiv}, each entry \(a_{pq}\) of the block
\(A_i^{+}\) is equivalent (under~\(C\)) to its rewritten form \(a_{pq}'\) in
\(\tilde{A}_i\).  Hence every entry of \(\hat{A}\) is provably equal to the
corresponding entry of \(\tilde{A}\), and therefore
\[
\kae
{\bigl(u^{(m)}\bigr)^{T}(\hat{A})^{*}v^{(m)}}
{
\bigl(u^{(m)}\bigr)^{T}(\tilde{A})^{*}v^{(m)}}.
\]
That is, \(\kae{\hat{p}}{\tilde{p}}\).
\end{proof}

\subsection{Equivalence}
Finally, we can state the equivalence. We first begin by the equivalence on words, then we go to the equivalence on the language.

\begin{lemma}[Commutation Equivalence via Factorization]
\label{thm:global-parikh}
Let $w_1, w_2 \in \Sigma^*$ with $C$-factorizations
\[
w_1 = w_{11} w_{12} \cdots w_{1n_1}
\quad\text{and}\quad
w_2 = w_{21} w_{22} \cdots w_{2n_2}.
\]
Then
\[
w_1 \equiv_C w_2
\quad\Longleftrightarrow\quad
n_1 = n_2
\;\text{ and }\;
\forall i \in [1,n_1],\;
\Psi(w_{1i}) = \Psi(w_{2i}).
\]
\end{lemma}

\begin{proof}
It is straightforward to show. Full details are delayed to the Appendix.
\end{proof}

With this equivalence on words in hand, we can reduce equivalence of expressions to equivalence of the corresponding languages.

\begin{theorem}\label{thm:coin-I-regular}
Let $\hat{p}$ and $\hat{q}$ be expressions obtained via factorization, and let
$\{S_1,\ldots,S_m\}$ be the supporting sets defined by
\[
S_i \;=\; \mathrm{Support}_i(\hat{p}) \,\cup\, \mathrm{Support}_i(\hat{q}),
\qquad 1 \le i \le m .
\]
Let $\tilde{p}$ and $\tilde{q}$ be the corresponding normal forms under these
supporting sets. Then
\[
\lc(\tilde{p}) = \lc(\tilde{q})
\quad\Longleftrightarrow\quad
L(\tilde{p}) = L(\tilde{q}).
\]
\end{theorem}

\begin{proof}
The implication $L(\tilde{p}) = L(\tilde{q}) \Rightarrow \lc(\tilde{p}) = \lc(\tilde{q})$ is trivial.

\smallskip
For the converse, assume \(\lc(\tilde{p}) = \lc(\tilde{q})\).
Let \(w \in L(\tilde{p})\) and let
\[
  w = w_1 w_2 \cdots w_n
\]
be its unique \(C\)-factorization.  
We show that \(w \in L(\tilde{q})\).

If \(w = 1\), then \(1 \in L(\tilde{p})\).  
Since \(\lc(\tilde{p}) = \lc(\tilde{q})\), it follows that
\(1 \in L(\tilde{q})\) as well.

Now assume \(w \neq 1\).
Because \(\lc(\tilde{p}) = \lc(\tilde{q})\), there exists
\(w' \in L(\tilde{q})\) such that \(w' \equiv_C w\).
By Theorem~\ref{thm:global-parikh},
\[
w' = w_1' w_2' \cdots w_n',
\qquad
w_i, w_i' \in \Sigma_{k_i}^+,
\qquad
\Psi(w_{i}) = \Psi(w_{i}').
\]

Write
\[
\tilde{p} = (u^{(m)})^{T}(\tilde{A}_p)^{*} v^{(m)},
\qquad
\tilde{q} = (u^{(m)})^{T}(\tilde{A}_q)^{*} v^{(m)}.
\]
Since $\tilde{p}$ and $\tilde{q}$ are built from the same supporting sets, each
entry of $\tilde{A}_p$ and $\tilde{A}_q$ is $0$ or a sum of the same atoms
$q_1,\ldots,q_M$, and these atoms satisfy
\[
\lc(q_k)\cap \lc(q_{k'})=\emptyset \qquad (k\neq k').
\]

Because $u^{(m)}$ and $v^{(m)}$ are $0$--$1$ vectors, and each nonzero entry of
$\tilde{A}_p$ is a sum of atoms all drawn from the same sub-alphabet
$\Sigma_i$, and moreover every entry is $\epsilon$-free, we have
\[
w \in L\bigl((u^{(m)})^{T} (\tilde{A}_p)^{n} v^{(m)}\bigr).\]
Expanding the product $(u^{(m)})^{T}(\tilde{A}_p)^{n}v^{(m)}$ yields a unique
sequence of atoms
\(
q_1,\ldots,q_n
\)
such that
\[
w_i \in L(q_i)
\qquad\text{and}\qquad
L(q_1 q_2 \cdots q_n)\subseteq L(\tilde{p}).
\]

Similarly, from $w'\in L(\tilde{q})$ we obtain atoms
\(
q_{1}',\ldots,q_{n}'
\)
such that
\[
w_i' \in L(q_i')
\quad\text{and}\quad
L(q_1' q_2' \cdots q_n')\subseteq L(\tilde{q}).
\]

For each $i$, the words $w_i$ and $w_i'$ lie in the same class $\Sigma_{k_i}$.
Atoms inside $\Sigma_{k_i}$ are pairwise $C$-disjoint:
\[
\lc(q_k)\cap \lc(q_{k'})=\emptyset\qquad (k\neq k').
\]
Since $w_i \equiv_{C} w_i'$ (by $P(w_i)=P(w_i')$ and Theorem \ref{thm:cka-completeness}) and $w_i\in L(q_i)$, $w_i'\in L(q_i')$, the only
possible atom is the same one, hence $q_i = q_i'$ for all $i$.

Thus
\[
w \in L(q_1\cdots q_n)
   = L(q_1'\cdots q_n')
   \subseteq L(\tilde{q}),
\]
so $L(\tilde{p}) \subseteq L(\tilde{q})$.
The reverse inclusion is symmetric, hence
$L(\tilde{p}) = L(\tilde{q})$.
\end{proof}

Now that all the necessary lemmas are in place, we can conclude the desired coincidence result.

\begin{theorem}[Coincidence of $KA+C$ and $KA^{*}+C$]\label{thm:coin}
Let $C$ be a transitive commutativity condition.
Then for all expressions $p,q$,
\[
\ka{KA}{C}{p}{q}
\quad\Longleftrightarrow\quad
\ka{KA^{*}}{C}{p}{q}.
\]
\end{theorem}

\begin{proof}
The forward direction
\[
\ka{KA}{C}{p}{q}
\;\Longrightarrow\;
\ka{KA^{*}}{C}{p}{q}
\]
is immediate.

\smallskip
For the converse, assume
\[
\ka{KA^{*}}{C}{p}{q}.
\]
By Theorem~\ref{thm:equiv}, this is equivalent to
\[
\lc(p) = \lc(q).
\]

Let $\hat{p}$ and $\hat{q}$ be the factorized forms of $p$ and $q$,
and let $\{S_1,\ldots,S_m\}$ be the supporting sets defined by
\[
S_i \;=\; \mathrm{Support}_i(\hat{p}) \cup \mathrm{Support}_i(\hat{q}),
\qquad 1\le i\le m.
\]
Let $\tilde{p}$ and $\tilde{q}$ be the corresponding normal forms.
Since factorization and normalization preserve $C$-equivalence,
\[
\lc(\tilde{p}) = \lc(p) = \lc(q) = \lc(\tilde{q}).
\]

By Theorem~\ref{thm:coin-I-regular},
\[
L(\tilde{p}) = L(\tilde{q}).
\]
By completeness of Kleene algebra~\cite{kozen1994ka},
\[
\kae{\tilde{p}}{\tilde{q}},
\]

Finally, since factorization and normalization are provably correct
(Theorems~\ref{thm:ka-factorization} and~\ref{thm:supp-equiv}), we have
\[
\kae{p=\hat{p}=\tilde{p}}{\tilde{q}=\hat{q}=q},
\]
This completes the proof.
\end{proof}

\begin{corollary}[Decidability of Equivalence]
The equivalence problem
\(
\ka{KA}{C}{e_1}{e_2}
\)
is decidable whenever \(C\) is transitive.
\end{corollary}

\begin{proof}
This follows immediately by combining the coincidence theorem, Theorem~\ref{thm:coin}, between \(KA+C\) and \(KA^*+C\), Kozen's correspondence (Theorem \ref{thm:equiv}) between \(KA^*+C\) and regular languages, and the decision procedure for regular languages under transitive commutativity conditions~\cite{regtrans82}.
\end{proof}
\section{Undecidability of Universality}\label{sec:undec}

Prior work~\cite{kuzuetsov23kac} has shown that equivalence in Kleene algebra is undecidable over the four-letter alphabet
\(\{a,b,c,d\}\) under the commutativity condition
\[
\{(a,c),(a,d),(b,c),(b,d)\}.
\]
while \citet{azevedo25kac} considers an even larger commutativity condition of a similar form.

In this section, we strengthen this result in two directions. First, we show that undecidability already arises in the \emph{minimal non-transitive} setting. Second, we show undecidability of the \emph{universality problem}, that is, whether a given expression $e$ is equivalent to the universal language \(\Sigma^*\). More specifically, we prove that universality for Kleene algebra remains undecidable under the commutativity condition
\[
C=\{(a,b),(b,c)\}.
\]

Our proof follows a structure similar to that of \citet{kuzuetsov23kac}, combined with a modified construction inspired by \citet{Ibarra78} and the coincidence Theorem~\ref{thm:coin} from the previous section.

The key idea is to show that reasoning from commutativity conditions over all Kleene algebras, not necessarily $*$-continuous ones, is still expressive enough to encode a simple form of non-halting behavior of Turing machines.  For this purpose, following \citet{kuzuetsov23kac}, we use the notion of \emph{$c$-looping}.

Intuitively, a Turing machine is \emph{$c$-looping} if it has a special \emph{capturing} state $c$ such that, once the machine enters $c$, it remains there forever. Thus, $c$ is not a halting state, and reaching $c$ guarantees non-halting.
We construct a regular expression $T(M)$ such that
\begin{itemize}
    \item if the machine $M$ is $c$-looping, then \(\kae{T(M)}{\Sigma^*}\);
    \item if \(\ka{KA^*}{C}{T(M)}{\Sigma^*}\), then $M$ does not halt.
\end{itemize}

The undecidability result then follows from the fact that the sets
\[
\{M \mid M \text{ is $c$-looping}\}
\qquad\text{and}\qquad
\{M \mid M \text{ halts}\}
\]
are \textit{recursively inseparable}: there is no decidable set that contains the first set and is disjoint from the second. From the two properties above, the set
\[
\{M \mid \kae{T(M)}{\Sigma^*}\}
\]
cannot be decidable, because it separates these sets. Therefore, universality under commutativity conditions over the class of all Kleene algebras is undecidable.

\subsection{Definition}

Before presenting the proof, we first fix the necessary definitions. We consider a deterministic Turing machine with a semi-infinite tape,
\[
M=(Q,\Gamma,\delta,q_0,q_f),
\]
where \(Q\) is a finite set of states, \(\Gamma=\{0,1,B\}\) is the tape alphabet (with \(B\) denoting the blank symbol), \(q_0\in Q\) is the initial state, and \(q_f\in Q\) is the halting state with \(q_0\neq q_f\). The tape head is initially positioned at the left endpoint of the tape. The transition function
\[
\delta:(Q\setminus\{q_f\})\times \Gamma \to Q\times \{0,1\}\times \{\mathit{Left},\mathit{Right},\mathit{Stay}\}
\]
specifies that whenever the machine is in a non-halting state and reads a tape symbol in \(\Gamma\), it deterministicly moves to a next state, writes either \(0\) or \(1\) on the current tape cell, and then either moves the head left, moves it right, or keeps it in place.

In addition, we distinguish a special \emph{capturing state} \(q_c\in Q\) with the property that, once the machine enters \(q_c\), it remains there forever without moving the head. Concretely, we assume that the transition function on \(q_c\) is given by
\[
\delta(q_c,\sigma)=(q_c,0,\mathit{Stay})
\qquad\text{for every }\sigma\in\Gamma,
\]
so that once the machine reaches the capturing state, it never leaves it and never moves the head.

Without loss of generality, we assume that \(M\) never overwrites a tape symbol with the blank symbol \(B\). We also assume that the machine starts with an empty tape, that is, every tape cell initially contains the blank symbol \(B\), and that the tape head is initially positioned at the left endpoint of the tape.

Under these conventions, every configuration of \(M\) can be written in the form
\[
xqyB,
\]
where \(x,y\in\{0,1\}^*\) and \(q\in Q\). Intuitively, this means that the tape contains the finite non-blank string \(xy\), while all tape cells to the right of this segment contain the blank symbol \(B\). The machine is currently in state \(q\), and the tape head is positioned between \(x\) and \(y\), scanning the first symbol of \(y\). Thus, \(x\) denotes the tape content strictly to the left of the head, while \(y\) denotes the tape content at and to the right of the head.

Here we adapt the construction of \citet{Ibarra78} from general Turing machines to Turing machines with semi-infinite tapes. At a technical level, our proof requires a procedure for checking whether two configurations \(xqyB\) and \(x'q'y'B\) form a valid single-step transition of the machine. For a semi-infinite tape, given a configuration \(xqyB\), the relative position of each tape symbol in \(x\) and \(y\) is always well defined, since the left endpoint never moves. This makes the case analysis comparatively manageable.

By contrast, for a general Turing machine with a bi-infinite tape, there is no fixed left endpoint. In that setting, each step requires an additional case analysis to determine whether a new leftmost symbol has been created, which further splits into subcases. Given the formidable complexity of the final construction, as the reader will see later in Section~\ref{subsec:halting} and \ref{subsec:c-looping}, it is unclear whether such an argument would remain provable, or even humanly manageable, in \(KA+C\). To make the technical development accessible and to simplify the theoretical treatment, we therefore work instead with semi-infinite tapes.

\subsection{Encoding Halting as a Regular Trace Language}\label{subsec:halting}

In this section, we encode the halting problem as a regular trace language. Rather than working directly over the minimal non-transitive commutativity condition \(C=\{(a,b),(b,c)\}\), we first use a larger alphabet for readability and for easier theoretical development. This encoding will later be translated into a minimal commutativity condition \(C=\{(a,b),(b,c)\}\) in Section~\ref{subsec:undec}. In that later encoding, we keep the letter \(b\) fixed and use the letters \(a\) and \(c\) to encode all remaining symbols.

For now, let the alphabet be
\[
\Sigma=\{0,1,B,\#\}\cup Q\cup\{b\},
\]
equipped with the commutativity condition that \(b\) commutes with every other symbol, while no other pair of symbols commutes:
\[
C=\{(\sigma,b)\mid \sigma\in \{0,1,B,\#\}\cup Q\}.
\]

For convenience, we adopt the following terminology:
\begin{itemize}
    \item a \emph{word} is an element of \(\Sigma^*\);
    \item a \emph{trace} is an element of \(\Sigma^*/C\);
    \item an \emph{execution} is a trace consisting of a sequence of TM configurations that starts from an initial configuration and ends in a halting configuration.
\end{itemize}

Our encoding begins with the following set of traces, equipped with an auxiliary counter symbol \(b\) that counts the number of non-\(b\) symbols in a trace:
\[
\begin{aligned}
R \defeq \{\, xb^{|x|} \mid\ &x=\#ID_0\#\cdots\#ID_k,\\
&ID_0,\ldots,ID_k \text{ are configurations of } M,\\
&ID_0 \text{ is the initial configuration of } M,\\
&ID_k \text{ is a halting configuration of } M\,\}.
\end{aligned}
\]
The non-\(b\) part of such a trace may look like a valid execution, except that we have not yet checked whether every consecutive pair of configurations \(ID_i\) and \(ID_{i+1}\) forms a valid single-step transition. We therefore define \(H(M)\subseteq R\) to be the subset consisting of those traces that correspond to genuine halting executions of \(M\).

Our goal in this section is to recognize the complement of \(H(M)\), and ultimately to show that it satisfies the universality property
\[
\ka{KA^*}{C}{\overline{H}(M)}{\Sigma^*},
\]
since \(H(M)\) should be empty in the cases of interest.

The construction proceeds as follows. We define a regular expression
\[
\overline{H}(M) \defeq e_{\mathit{illegal}} + e_{\mathit{invalid\_execution}}(M) + \Sigma^* q_c \Sigma^*
\]
to recognize the complement of \(H(M)\), where
\begin{itemize}
    \item \(e_{\mathit{illegal}}\) accepts the complement of \(R\);
    \item \(e_{\mathit{invalid\_execution}}(M)\) accepts a superset of the traces in \(R\setminus H(M)\);
    \item \(\Sigma^* q_c \Sigma^*\) denotes those executions that enter the capturing state somewhere along the trace, and therefore do not halt. This term is technically redundant in the present section when proving
    \[
    \ka{KA^*}{C}{\overline{H}(M)}{\Sigma^*},
    \]
    since any finite execution containing \(q_c\) is already invalid. However, we can show this redundancy only in \(KA^*+C\), not in \(KA+C\). For this reason, we keep the term explicitly, as it will play an essential role in the next section, where we will additionally show that
    \[
    KA + C \vDash \overline{H}(M) = \Sigma^*
    \]
    whenever \(M\) is \(c\)-looping.
\end{itemize}

\subsubsection{Encoding of $e_{\mathit{illegal}}$} We first begin with the encoding of expression $e_{\mathit{illegal}}$.
For convenience, write
\[
\Sigma_{\setminus b} \defeq 0+1+B+\#+\sum_{q\in Q}q.
\]

We begin by separating those words in which the number of occurrences of \(b\) does not match the number of all other symbols.

To this end, define
\[
e_{\mathit{mismatch\_b}}
\defeq
b^+(\Sigma_{\setminus b}b)^*
\;+\;
\Sigma_{\setminus b}\bigl(\Sigma_{\setminus b}(1+b)\bigr)^*,
\]
and
\[
e_{\mathit{match\_b}} \defeq (\Sigma_{\setminus b}b)^*.
\]
Intuitively, \(e_{\mathit{match\_b}}\) describes those words in which each symbol from \(\Sigma_{\setminus b}\) is paired with exactly one \(b\), while \(e_{\mathit{mismatch\_b}}\) describes the remaining words, in which the numbers do not match.
Now that we have separated out all words in which the number of occurrences of $b$ does not match the number of non-$b$ symbols, the next step is to separate those words that do not encode an execution trace. The language of correctly formed configurations and traces is described by the following regular expressions. Here, $e_{ID}$ represents a single configuration, and $e_{\mathit{traces}}$ represents traces that begin with the initial configuration and end with a halting configuration:
\[
\begin{aligned}
e_{ID}
&=
\#b(0b+1b)^*\Bigl(\sum_{q\in Q} qb\Bigr)(0b+1b)^*Bb,\\
e_{\mathit{traces}}
&=
\#bq_0bBb(e_{ID})^*
\#b(0b+1b)^*q_fb(0b+1b)^*Bb.
\end{aligned}
\]

Notice that, at the level of ordinary regular languages, that is, without commutativity conditions, one can effectively compute a regular expression $e_1\setminus e_2$ such that
\[
L(e_1\setminus e_2)=L(e_1)\setminus L(e_2).
\]
We therefore define
\[
e_{\mathit{illegal}} = e_{\mathit{mismatch\_b}}+ e_{\mathit{match\_b}} \setminus e_{\mathit{traces}}.
\]

\subsubsection{Encoding of $e_{\mathit{invalid\_execution}}(M)$}
Now that we can isolate all syntactically valid traces using $e_{\mathit{traces}}$, the final step is to identify those traces that do not encode a \emph{legal} execution of the Turing machine. In other words, we want to detect those traces in $e_{\mathit{traces}}$ that contain two consecutive configurations $ID_i$ and $ID_{i+1}$ such that $ID_{i+1}$ is not a valid successor of $ID_i$.

To this end, we define an expression $e_{\mathit{invalid\_IDs}}(M)$ intended to capture illegal pairs of consecutive configurations. Ideally, we would like to recognize the language
\[
\{\,\#ID_1\#ID_2 b^{|\#ID_1\#ID_2|}\mid ID_2 \text{ is not a valid successor of } ID_1\,\}.
\]
However, this language is not regular, since checking validity of the transition requires comparing arbitrarily long configurations. This can be shown formally using the pumping lemma.

The extra symbol $b$ allows us to bypass this difficulty. We construct a regular expression $e_{\mathit{invalid\_IDs}}(M)$ such that, whenever a trace belongs to $L_C(e_{\mathit{invalid\_IDs}}(M))$ and the number of occurrences of $b$ is equal to the number of all other symbols, the word must be of the form
\[
\#ID_1\#ID_2 b^{|\#ID_1\#ID_2|},
\]
and the pair $\#ID_1\#ID_2$ is necessarily illegal. For words of the form
\[
\#ID_1\#ID_2 b^n
\qquad\text{with}\qquad
n\neq |\#ID_1\#ID_2|,
\]
their membership is irrelevant.

Following the idea of \citet{Ibarra78}, every illegal transition from $ID_1$ to $ID_2$ can be recognized by inspecting corresponding windows of three consecutive symbols in $ID_1$ and $ID_2$, that is, the $s$th, $(s+1)$st, and $(s+2)$nd symbols. For example, if there is a copying error on the tape away from $\#$, $q$, and $B$, then there exists some position $s$ such that
\[
ID_1[s,s+2]\neq ID_2[s,s+2],
\]
with all symbols involved belonging to $\{0,1\}$. Similarly, if there is an error in the machine transition, then there exists some position $s$ such that the local windows $ID_1[s,s+2]$ and $ID_2[s,s+2]$ contain a state symbol and together represent an invalid transition. After a detailed case analysis, given in the Appendix, we obtain two finite sets:
\[
\Delta\subseteq \Sigma^6,
\qquad
\Delta'\subseteq \Sigma^6.
\]
Here, $\Delta$ consists of all invalid corresponding $6$-tuples that do not involve a blank symbol in the local window, while $\Delta'$ consists of those in which a blank symbol is involved.

We now translate these observations into a regular expression. Let
\[
\Sigma_{ID}=0+1+B+\sum_{q\in Q} q
\qquad\text{and}\qquad
\Sigma_{ID\setminus B}=0+1+\sum_{q\in Q} q
\]
denote, respectively, the set of symbols that may appear in a configuration and the subset of those symbols excluding the blank symbol $B$.

We define
\[
\begin{aligned}
e_{\mathit{invalid\_IDs}}(M)
={}&\ e_{\mathit{illegal}<3}(M)\\
&+\#b^2(\Sigma_{ID}b^2)^*
\Biggl(
\sum_{(l_1,l_2,l_3,l_4,l_5,l_6)\in \Delta}
\Bigl(
l_1bl_2bl_3b(\Sigma_{ID}b)^*\#(\Sigma_{ID})^*l_4bl_5bl_6b
\Bigr)
\Biggr)
(\Sigma_{ID}b)^*\\
&+\#b^2(\Sigma_{ID}b^2)^*
\Biggl(
\sum_{(l_1,l_2,l_3,l_4,l_5,l_6)\in \Delta'}
\Bigl(
l_1bl_2bl_3b(\Sigma_{ID}b)^*\#(\Sigma_{ID})^*l_4bl_5bl_6b
\Bigr)
\Biggr)
\Bigl(\epsilon+\Sigma_{ID\setminus B}b(\Sigma_{ID}b)^*\Bigr).
\end{aligned}
\]

Here, the first summand $e_{\mathit{illegal}<3}(M)$ covers the exceptional short-configuration cases, where one of the configurations has tape length less than $2$. The second summand captures all pairs of configurations containing a local illegal pattern from $\Delta$, that is, patterns whose local window does not involve the blank symbol $B$. The third summand handles the remaining patterns in $\Delta'$, where a blank symbol occurs in the local window.

The key idea is to guess a position in the two consecutive configurations $ID_1$ and $ID_2$, and then use the number of occurrences of $b$ to verify that the guessed positions are aligned. More precisely, before guessing the beginning of an illegal local pattern from $\Delta$ or $\Delta'$, we move through the word at speed $2b$ per symbol of $\#\cup\Sigma_{ID}$. Once we guess the first symbol $l_1$ of such a pattern, we slow down and consume exactly one $b$ per symbol until we reach the separator $\#$ marking the beginning of $ID_2$. After that, we stop consuming $b$'s until we guess the corresponding symbol $l_4$ in $ID_2$.

As a result, the guessed symbols $l_1$ and $l_4$ occur at the same position in $ID_1$ and $ID_2$ if and only if the final word contains exactly as many occurrences of $b$ as of all other symbols. Thus, equality in the number of $b$'s certifies that the guessed local mismatch occurs at corresponding positions in the two configurations.

We therefore define
\[
e_{\mathit{invalid\_execution}}(M)
=
e_{ID}^*e_{\mathit{invalid\_IDs}}(M)e_{ID}^*.
\]
Our final expression is
$\overline{H}(M) \defeq e_{\mathit{illegal}} + e_{\mathit{invalid\_exection}}(M) + \Sigma^* q_c \Sigma^*$
\begin{theorem}[Halting checking]\label{thm:halting-raw}
Let
\(
M=(Q,\Gamma,\delta,q_0,q_f)
\)
be a Turing machine on a semi-infinite tape that never overwrites any symbol by the blank symbol and has a capturing state $q_c$. Let $e_{\mathit{translate}}$ be the expression constructed above from $M$.
\(
\ka{KA^*}{C}{\overline{H}(M)}{\Sigma^*}
\) under \(C=\{(\sigma,b)\mid \sigma\in \{0,1,B,\#\}\cup Q\}\)
iff $M$ does not halt.
\end{theorem}

\begin{proof}
Since equivalence in $KA^*+C$ coincides with equality of the corresponding regular languages, it suffices to reason at the level of regular languages. See full proof in the Appendix.
\end{proof}

\subsection{Encoding $c$-looping}\label{subsec:c-looping}

In the previous section, we defined the key expression as
\[
\overline{H}(M) \defeq e_{\mathit{illegal}} + e_{\mathit{invalid\_execution}}(M) + \Sigma^* q_c \Sigma^*,
\]
and included the term \(\Sigma^* q_c \Sigma^*\), even though it did not appear to play an essential role there. In this section, this term becomes crucial.

The main difficulty in proving results in \(KA+C\) is that we cannot freely unfold the star operator, but must instead use it only through the axioms of Kleene algebra. In fact, for the theorem established in the previous section,
\[
KA^* + C \vDash \overline{H}(M) = \Sigma^*
\qquad\text{iff}\qquad
M \text{ halts},
\]
it is not clear whether the analogous statement remains true, or can even be proved, for \(KA+C\), because the size of the encoding and the complexity of reasoning about the star operator already exceed what is humanly manageable. We therefore take one step back and prove a weaker but sufficient statement: 
\[
\text{If \(M\) is \(c\)-looping,  then } KA + C \vDash \overline{H}(M)=\Sigma^*.
\]

The intuition behind the \(c\)-looping argument is simple. Since the Turing machines we consider are deterministic, if \(M\) is \(c\)-looping, then there exists a finite number \(n\) such that after \(n\) steps of execution the machine enters the capturing state \(q_c\). Therefore, for any trace that could potentially encode an execution, one of two things must happen: either (1) the trace contains the capturing state \(q_c\), or (2) an execution error already occurs within the first \(n\) steps.

In general, it is difficult to prove in \(KA+C\) that an execution error occurs somewhere along an arbitrarily long trace, because we cannot unfold the star operator at will. However, once we know that any error must occur within the first \(n\) steps, the argument becomes much more manageable.

Moreover, since our Turing machine can increase the tape length by at most one symbol at each step, a trace containing \(n\) steps of execution can be further divided into two cases: either (1) there are two consecutive configurations whose lengths differ by at least \(2\), or (2) the trace has bounded finite length. We handle these two cases separately in the proof.

We begin by showing that the basic construction for \(KA^*+C\) also holds for \(KA+C\).
\begin{lemma}[Legal and Illegal Traces]
\label{thm:trace}
Let
\[
\Sigma=\{0,1,B,\#\}\cup Q\cup\{b\},
\qquad
C=\{(\sigma,b)\mid \sigma\in \{0,1,B,\#\}\cup Q\}.
\]
Then
\[
\ka{KA}{C}{e_{\mathit{traces}}+e_{\mathit{illegal}}}{\Sigma^*}.
\]
\end{lemma}

\begin{proof}
The proof of this lemma is somewhat technical, since in \(KA\) the use of the star operator is restricted to the axioms of Kleene algebra. We simplify part of the argument by appealing to the coincidence Theorem~\ref{thm:coin}, using the fact that \(\Sigma_{\setminus b}\) and \(b\) commute, and that the expressions \(e_{\mathit{match\_b}}\) and \(e_{\mathit{mismatch\_b}}\) can be encoded in terms of \(\Sigma_{\setminus b}\) and \(b\), rather than treated directly over the full alphabet \(\Sigma\). The overall argument is nevertheless straightforward; see the Appendix for details.\end{proof}

We next show that case~(1), namely the case in which two consecutive configurations differ in length by at least \(2\), is always captured by \(e_{\mathit{invalid\_execution}}(M)\). It therefore remains only to consider traces of bounded finite length.

\begin{lemma}[Case Analysis]\label{thm:case-analysis}
Let
\[
M=(Q,\Gamma,\delta,q_0,q_f)
\]
be a Turing machine on a semi-infinite tape that never overwrites any symbol with the blank symbol and has a capturing state \(q_c\). Write
\[
e_{\mathit{init}}=\#bq_0bBb,
\qquad
e_{ID}=\#b(0b+1b)^*\Bigl(\sum_{q\in Q}qb\Bigr)(0b+1b)^*Bb.
\]
Then for every \(n\in\mathbb{N}\),
\[
KA + C \vDash
e_{\mathit{init}}e_{ID}^{n}
\le
e_{\mathit{invalid\_execution}}(M)
+
\sum_{\substack{|w|\le \frac{n^2+7n+6}{2}\\ w\in L(e_{\mathit{init}}e_{ID}^{n})}} w.
\]
\end{lemma}

\begin{proof}
Here \(e_{\mathit{init}}\) denotes the expression for the initial configuration of the Turing machine, while \(e_{ID}\) denotes the expression for a single intermediate configuration. Thus, the expression \(e_{\mathit{init}}e_{ID}^{n}\) represents the first \(n\) steps of a potential execution of the machine.

A detailed technical argument, given in the Appendix, shows that traces of the first kind, namely those in which some consecutive pair of configurations differs in length by at least \(2\), are always included in \(e_{\mathit{invalid\_execution}}(M)\). It therefore remains to bound the length of the remaining traces. Since every configuration has length at least \(3\), accounting for \(\#\), a state symbol \(q\), and the blank symbol \(B\), and since each step increases the length by at most \(1\), every such trace of \(n\) steps has total length at most
\[
\sum_{k=0}^{n}(k+3)=\frac{n^2+7n+6}{2}.
\]
The lemma then follows. We defer the full details to the Appendix.
\end{proof}

With Lemmas~\ref{thm:trace} and~\ref{thm:case-analysis} in hand, our \(c\)-loop checking theorem follows straightforwardly.

\begin{theorem}[$c$-loop checking]\label{thm:c-loop-raw}
Let
\(
M=(Q,\Gamma,\delta,q_0,q_f)
\)
be a Turing machine on a semi-infinite tape, which never overwrites any symbol by the blank symbol and has a capturing state $q_c$. Let $\overline{H}(M) $ be the expression constructed above from $M$. If $M$ is $c$-looping, then
\(
\ka{KA}{C}{\overline{H}(M) }{\Sigma^*}
\) under \(C=\{(\sigma,b)\mid \sigma\in \{0,1,B,\#\}\cup Q\}.\)
\end{theorem}

\begin{proof}
It suffices to prove
\[
\kale{KA}{C}{\overline{H}(M) }{\Sigma^*}
\qquad\text{and}\qquad
\kale{KA}{C}{\Sigma^*}{\overline{H}(M) }.
\]

The first inequality is immediate: since
\[
\karle{\overline{H}(M) }{\Sigma^*},
\]
we also have
\[
\kale{KA}{C}{\overline{H}(M) }{\Sigma^*}.
\]

For the converse direction, by Theorem~\ref{thm:trace}, we have
\[
\kale{KA}{C}{\Sigma^*}{e_{\mathit{illegal}}+e_{\mathit{traces}}}.
\]
Since
\[
\overline{H}(M) 
=
e_{\mathit{illegal}}
+
e_{\mathit{invalid\_execution}}(M)
+
\Sigma^*q_c\Sigma^*,
\]
it remains to show that
\[
\kale{KA}{C}{e_{\mathit{traces}}}{e_{\mathit{invalid\_execution}}(M)+\Sigma^*q_c\Sigma^*}.
\]

The rest is a case analysis. Since \(M\) is \(c\)-looping, there exists \(n\in\mathbb{N}\) such that every trace containing more than \(n\) configurations is either invalid or contains the state \(q_c\). The desired inclusion then follows from Lemma~\ref{thm:case-analysis}, together with a finite case analysis. Full details are given in the Appendix.
\end{proof}

\subsection{Undecidability}\label{subsec:undec}
Finally, we prove undecidability. The argument is based on recursive inseparability.

\begin{definition}[Recursively inseparable]
Let $A,B\subseteq \mathbb{N}$ be two disjoint sets. We say that $A$ and $B$ are \emph{recursively inseparable} if there is no decidable set $S\subseteq \mathbb{N}$ such that
\[
A\subseteq S
\qquad\text{and}\qquad
B\cap S=\varnothing.
\]
Equivalently, there is no decidable set that contains all elements of $A$ while excluding all elements of $B$.
\end{definition}

In our setting, recursive inseparability is obtained by a short argument from the standard recursive inseparability result for general Turing machines, which we also include in the Appendix.

\begin{lemma}[Recursive inseparability of $c$-looping and halting]
For the class of Turing machines \(M\) that operate on a semi-infinite tape, have a capturing state \(q_c\), and never overwrite a tape symbol with a blank symbol, the sets
\[
\{\,M \mid M \text{ is \(c\)-looping}\,\}
\qquad\text{and}\qquad
\{\,M \mid M \text{ halts}\,\}
\]
are recursively inseparable.
\end{lemma}

\begin{proof}
See Appendix.
\end{proof}

With the recursive inseparability in place, we only need one more step before getting the final theorem, that is we will convert the above result into a minimal encoding.
\begin{lemma}[Minimal Encoding]\label{thm:min-enc}
Let
\[
M=(Q,\Gamma,\delta,q_0,q_f)
\]
be a Turing machine on a semi-infinite tape with a capturing state \(q_c\), and suppose that \(M\) never overwrites any symbol with a blank symbol. Then there exists an effective translation \(T\) such that \(T(M)\) is a regular expression over an arbitrary alphabet \(\Sigma\) containing at least the three letters \(\{a,b,c\}\), equipped with a commutativity condition \(C\) such that
\[
(a,b)\in C,\qquad (b,c)\in C,\qquad (a,c)\notin C.
\]
The following hold:
\begin{itemize}
    \item if \(M\) is \(c\)-looping, then
    \[
    \ka{KA}{C}{T(M)}{\Sigma^*};
    \]
    \item if
    \[
    \ka{KA^*}{C}{T(M)}{\Sigma^*},
    \]
    then \(M\) does not halt.
\end{itemize}
\end{lemma}
\begin{proof}
In the previous construction, we defined a translation from \(M\) to a regular expression over the alphabet
\[
\Sigma'=\{0,1,B,\#\}\cup Q\cup\{b\},
\]
with commutativity condition
\[
C'=\{(\sigma,b)\mid \sigma\in \{0,1,B,\#\}\cup Q\}.
\]
We now reduce this alphabet to \(\{a,b,c\}\) by encoding each symbol in \(\{0,1,B,\#\}\cup Q\) by a distinct binary word of length
\[
\left\lceil \log_2\bigl(|\{0,1,B,\#\}\cup Q|\bigr)\right\rceil
\]
over \(\{a,c\}\), while keeping \(b\) unchanged. We denote this encoding by \(f\).

Since \(f\) is a binary encoding into \(\{a,c\}^*\) while keeping the $b$'s, not every word over \(\Sigma\) lies in the image of the translation. To prove universality, we therefore introduce a regular expression for the complement of the encoded image:
\[
e_{\mathit{complement}}
=
\Sigma^*\setminus \Bigl(\sum_{\sigma\in\Sigma'} f(\sigma)\Bigr)^*.
\]
We then define the final expression by
\[
T(M)=f\bigl(\overline{H}(M)\bigr)+e_{\mathit{complement}}.
\]

By construction, we immediately have
\[
\kar{e_{\mathit{complement}}+\Bigl(\sum_{\sigma\in\Sigma'} f(\sigma)\Bigr)^*}{\Sigma^*}.
\]
The remaining argument is straightforward: one transfers the properties of \(\overline{H}(M)\) through the encoding \(f\), and then uses the complement term to cover all words outside the image of the encoding. We defer the full details to the Appendix.
\end{proof}

We can now conclude the undecidability result.
\begin{theorem}[Undecidability of Universality]
\label{thm:undec-main}
The universality problem
\(
\ka{KA}{C}{e}{\Sigma^*}
\)
is undecidable whenever \(C\) is not transitive.
\end{theorem}

\begin{proof}
Since \(C\) is not transitive, there exist letters \(a,b,c\in\Sigma\) such that
\[
(a,b)\in C,\qquad (b,c)\in C,\qquad (a,c)\notin C.
\]
By Lemma~\ref{thm:min-enc}, we can effectively construct, from each Turing machine \(M\), an expression \(T(M)\) such that:
\begin{itemize}
    \item if \(M\) is \(c\)-looping, then
    \[
    \ka{KA}{C}{T(M)}{\Sigma^*};
    \]

    \item if \(M\) halts, then
    \[
    KA+C\not\vDash T(M)=\Sigma^*.
    \]
    Indeed, suppose towards a contradiction that
    \[
    \ka{KA}{C}{T(M)}{\Sigma^*}.
    \]
    Since every equation valid in \(KA\) is also valid in \(KA^*\), it follows that
    \[
    \ka{KA^*}{C}{T(M)}{\Sigma^*}.
    \]
    By Lemma~\ref{thm:min-enc}, this implies that \(M\) does not halt, a contradiction.
\end{itemize}

We now apply recursive inseparability. Assume, for contradiction, that the predicate
\[
\ka{KA}{C}{T(M)}{\Sigma^*}
\]
is decidable. Then the set
\[
S=\{\,M \mid \ka{KA}{C}{T(M)}{\Sigma^*}\,\}
\]
would be decidable.

Therefore, \(S\) is a decidable set such that
\[
\{\,M \mid M \text{ is \(c\)-looping}\,\}\subseteq S
\qquad\text{and}\qquad
S\cap \{\,M \mid M \text{ halts}\,\}=\varnothing.
\]
This contradicts the recursive inseparability of the sets
\[
\{\,M \mid M \text{ is \(c\)-looping}\,\}
\qquad\text{and}\qquad
\{\,M \mid M \text{ halts}\,\}.
\]
Hence the universality problem is undecidable.
\end{proof}
\begin{corollary}[Undecidability of Equivalence]
The equivalence problem
\(
\ka{KA}{C}{e_1}{e_2}
\)
is undecidable whenever \(C\) is not transitive.
\end{corollary}

\begin{proof}
Take
\(
e_2=\Sigma^*.
\)
Then the equivalence problem specializes to the universality problem, which is undecidable by Theorem~\ref{thm:undec-main}.
\end{proof}
\section{Related Work}

In this section, we discuss related work and its role in the development of this line of research. It is already known that regular languages with commutativity conditions are decidable if and only if \(C\) is transitive~\cite{Ibarra78,regtrans82}. The constructions in these works are relatively concise, each occupying only about half a page. Our undecidability proof uses a modified version of the construction of \citet{Ibarra78}, while our decidability proof shares the idea of partitioning the alphabet into subalphabets of transitively commuting letter groups with the earlier decision procedure of \citet{regtrans82}. However, the factorization process for constructing an \(\epsilon\)-free matrix, together with the subsequent supporting-expression framework, is new to this work.

The main difficulty in our setting, and more broadly throughout the line of other results for Kleene Algebra~\cite{kozen1994ka,Kozen96KAT,kozen02comp,kuzuetsov23kac,azevedo25kac}, does not lie in the underlying decidability or undecidability constructions, but rather in proving the corresponding results while using the \( * \)-operator only in ways justified by the axioms of Kleene algebra.
For regular languages, one can readily unfold the star operation as an infinite sum
\[
e^*=\sum_i e^i.
\]
This makes it possible to reason about \(e^*\) through its finite approximants \(e^i\): properties of \(e^*\) can often be established by proving the corresponding statement for each \(e^i\), typically by induction on \(i\). In Kleene algebra, however, one cannot literally decompose a star expression into smaller pieces in this way. Instead, one must reason algebraically through the star axioms, for example
\[
ab \le b \Rightarrow a^*b \le b
\qquad\text{and}\qquad
1+aa^* \le a^*.
\]
This is the main reason why the transfer from regular-language arguments to \(KA+C\) is substantially more subtle.

Thus, more than forty years after the regular-language result was proposed, undecidability of \(KA+C\) was finally established independently by \citet{kuzuetsov23kac} and \citet{azevedo25kac}. The proof of \citet{kuzuetsov23kac} encodes undecidability via Post's correspondence problem, following \citet{Kozen96KAT}'s original construction for proving undecidability of \(KA^*+C\), while \citet{azevedo25kac} uses encoding from two-counter machines. However, PCP essentially requires two sets of words, while two-counter machines require separate encodings for the two counters. Together with the commutativity patterns needed for each set of words or counters, these constructions fundamentally require an alphabet with four letters, say \(\{a,b,c,d\}\), together with a non-transitive commutativity pattern such as
\[
\{(a,c),(a,d),(b,c),(b,d)\}.
\]
In contrast, we bypass this limitation by encoding the problem directly from Turing machines, rather than through the intermediate encodings of PCP or two-counter machines. This allows us to strengthen the previous results both in terms of the minimal non-transitive commutativity conditions and in showing that universality is already undecidable.

On the decidability side, another line of related work worth mentioning concerns coincidence results between \(KA\) and \(KA^*\), as well as some of their extensions. Our gesture that the equational theory of \(KA+C\) should coincide with that of \(KA^*+C\) is motivated by the coincidence of \(KA\) with \(KA^*\)~\cite{kozen1994ka}, as well as the coincidence of \(CKA\) with \(CKA^*\)~\cite{pilling1970algebra}. Still, such coincidence theorems typically require substantially different proofs in different settings. Kozen's proof for \(KA\)~\cite{kozen1994ka} relies on a direct automata-theoretic encoding, while the corresponding result for \(CKA\) is proved by showing that its equational theory coincides with equality of Parikh images. The results for \(KAT\)~\cite{kozen97kat} and for certain systems \(KA+E\), where \(E\) rewrites words of length greater than one into single letters~\cite{kozen14kae}, reduce the problem to ordinary \(KA\) equivalence, but again through case-specific arguments.

Our proof is inspired by the overall strategy used for \(KAT\)~\cite{kozen97kat} and \(KA+E\)~\cite{kozen14kae}: for each expression \(p\), we first construct a provably equivalent normal form \(\widetilde{p}\), and then show that two such normal forms are equivalent in \(KA+C\) if and only if they are equivalent in ordinary \(KA\). Still, as in other coincidence proofs, our construction of the normal form \(\widetilde{p}\) is entirely different from those used for \(KAT\) and \(KA+E\).

More specifically, the \(KA+E\) results apply to systems in which \(E\) rewrites words of length greater than one into single letters, which is fundamentally different from commutativity conditions, where each equation \(ab=ba\) rewrites a word of length \(2\) into another word of length \(2\). For \(KAT\), one might ask whether \(KA+C\) could be embedded directly into a \(KAT\) framework, since our construction fundamentally relies on the Boolean algebra of semilinear sets. However, the answer is negative. In \(KA+C\), the Parikh interpretation satisfies
\[
P(L(0))=\emptyset
\qquad\text{and}\qquad
P(L(1))=\{(0,0,\ldots,0)\},
\]
whereas in the Boolean algebra of semilinear sets used for KAT,
\[
0_B=\emptyset
\qquad\text{and}\qquad
1_B=\mathbb{N}^k.
\]
This mismatch prevents a direct embedding of \(KA+C\) into \(KAT\), and hence prevents us from reusing the existing constructions for \(KAT\).

Once the coincidence result is established, we can combine the existing regular language results~\cite{kozen02comp,regtrans82} with our undecidability result to obtain the full decidability landscape of Kleene algebra with commutativity conditions.
\section{Conclusion}

In this paper, we settle the decidability of $KA+C$ by identifying \emph{transitivity} of the commutativity conditions as the exact structural property that determines the decidability boundary. We further strengthen the picture in both directions: when $C$ is transitive, the equational theories of \kc\ and \kcs\ coincide; when $C$ is not transitive, the universality problem for \kc\ is undecidable. 

\bibliography{reference,more-references}

@inproceedings{netkat, author = {Anderson, Carolyn Jane and Foster, Nate and Guha, Arjun and Jeannin, Jean-Baptiste and Kozen, Dexter and Schlesinger, Cole and Walker, David}, title = {{NetKAT}: Semantic Foundations for Networks}, year = {2014}, booktitle = {ACM SIGPLAN-SIGACT Symposium on Principles of Programming Languages}, pages = {113–126} }

@article{kozen1994ka,
title = {A Completeness Theorem for Kleene Algebras and the Algebra of Regular Events},
journal = {Information and Computation},
volume = {110},
number = {2},
pages = {366-390},
year = {1994},
issn = {0890-5401},
doi = {https://doi.org/10.1006/inco.1994.1037},
url = {https://www.sciencedirect.com/science/article/pii/S0890540184710376},
author = {D. Kozen}
}

@book{Davey_Priestley_2002, place={Cambridge}, edition={2}, title={Introduction to Lattices and Order}, publisher={Cambridge University Press}, author={Davey, B. A. and Priestley, H. A.}, year={2002}}

@article{Halmos1966-HALLOB,
	author = {Paul R. Halmos},
	doi = {10.2307/2269816},
	journal = {Journal of Symbolic Logic},
	number = {2},
	pages = {253--254},
	publisher = {Association for Symbolic Logic},
	title = {Lectures on Boolean Algebras},
	volume = {31},
	year = {1966}
}

@book{Sikorski1960-SIKBA-2,
	address = {Berlin, Germany},
	author = {Roman Sikorski},
	editor = {},
	publisher = {Springer},
	title = {Boolean Algebras},
	year = {1960}
}

@phdthesis{pilling1970algebra,
  title={The algebra of operators for regular events},
  author={Pilling, Donald L},
  year={1970},
  school={Cambridge, UK; Cambridge University}
}

@inproceedings{to2010parikh,
author = {Kopczynski, Eryk and To, Anthony Widjaja},
title = {Parikh Images of Grammars: Complexity and Applications},
year = {2010},
isbn = {9780769541143},
publisher = {IEEE Computer Society},
address = {USA},
url = {https://doi.org/10.1109/LICS.2010.21},
doi = {10.1109/LICS.2010.21},
booktitle = {Proceedings of the 2010 25th Annual IEEE Symposium on Logic in Computer Science},
pages = {80–89},
numpages = {10},
series = {LICS '10}
}

@article{parikh1966context,
author = {Parikh, Rohit J.},
title = {On Context-Free Languages},
year = {1966},
issue_date = {Oct. 1966},
publisher = {Association for Computing Machinery},
address = {New York, NY, USA},
volume = {13},
number = {4},
issn = {0004-5411},
url = {https://doi.org/10.1145/321356.321364},
doi = {10.1145/321356.321364},
journal = {J. ACM},
month = oct,
pages = {570–581},
numpages = {12}
}

@inproceedings{badban2010semilinear,
author = {Badban, Bahareh and Dashti, Mohammad Torabi},
title = {Semi-linear Parikh images of regular expressions via reduction},
year = {2010},
isbn = {364215154X},
publisher = {Springer-Verlag},
address = {Berlin, Heidelberg},
booktitle = {Proceedings of the 35th International Conference on Mathematical Foundations of Computer Science},
pages = {653–664},
numpages = {12},
location = {Brno, Czech Republic},
series = {MFCS'10}
}

@article{ginsburg66semi,
author = {Seymour Ginsburg and Edwin H. Spanier},
title = {{Semigroups, Presburger formulas, and languages.}},
volume = {16},
journal = {Pacific Journal of Mathematics},
number = {2},
publisher = {Pacific Journal of Mathematics, A Non-profit Corporation},
pages = {285 -- 296},
year = {1966},
}

@InProceedings{kozen97kat,
author="Kozen, Dexter
and Smith, Frederick",
editor="van Dalen, Dirk
and Bezem, Marc",
title="Kleene algebra with tests: Completeness and decidability",
booktitle="Computer Science Logic",
year="1997",
publisher="Springer Berlin Heidelberg",
address="Berlin, Heidelberg",
pages="244--259",
isbn="978-3-540-69201-0"
}

@article{Brunet19note,
  author       = {Paul Brunet},
  title        = {A note on commutative Kleene algebra},
  journal      = {CoRR},
  volume       = {abs/1910.14381},
  year         = {2019},
  url          = {http://arxiv.org/abs/1910.14381},
  eprinttype    = {arXiv},
  eprint       = {1910.14381},
  bibsource    = {dblp computer science bibliography, https://dblp.org}
}

@InProceedings{kozen14kae,
author="Kozen, Dexter
and Mamouras, Konstantinos",
editor="Esparza, Javier
and Fraigniaud, Pierre
and Husfeldt, Thore
and Koutsoupias, Elias",
title="Kleene Algebra with Equations",
booktitle="Automata, Languages, and Programming",
year="2014",
publisher="Springer Berlin Heidelberg",
address="Berlin, Heidelberg",
pages="280--292",
isbn="978-3-662-43951-7"
}

@inproceedings{regtrans82,
author = {Bertoni, Alberto and Mauri, Giancarlo and Sabadini, Nicoletta},
title = {Equivalence and Membership Problems for Regular Trace Languages},
year = {1982},
isbn = {3540115765},
publisher = {Springer-Verlag},
address = {Berlin, Heidelberg},
booktitle = {Proceedings of the 9th Colloquium on Automata, Languages and Programming},
pages = {61–71},
numpages = {11}
}

@article{kozen02comp,
title = {On the Complexity of Reasoning in Kleene Algebra},
journal = {Information and Computation},
volume = {179},
number = {2},
pages = {152-162},
year = {2002},
issn = {0890-5401},
doi = {https://doi.org/10.1006/inco.2001.2960},
url = {https://www.sciencedirect.com/science/article/pii/S0890540101929608},
author = {Dexter Kozen}
}

@InProceedings{kuzuetsov23kac,
author="Kuznetsov, Stepan L.",
editor="{\'A}brah{\'a}m, Erika
and Dubslaff, Clemens
and Tarifa, Silvia Lizeth Tapia",
title="On the Complexity of Reasoning in Kleene Algebra with Commutativity Conditions",
booktitle="Theoretical Aspects of Computing -- ICTAC 2023",
year="2023",
publisher="Springer Nature Switzerland",
address="Cham",
pages="83--99",
isbn="978-3-031-47963-2"
}

@InProceedings{azevedo25kac,
  author =	{Azevedo de Amorim, Arthur and Zhang, Cheng and Gaboardi, Marco},
  title =	{{Kleene Algebra with Commutativity Conditions Is Undecidable}},
  booktitle =	{33rd EACSL Annual Conference on Computer Science Logic (CSL 2025)},
  pages =	{36:1--36:25},
  series =	{Leibniz International Proceedings in Informatics (LIPIcs)},
  ISBN =	{978-3-95977-362-1},
  ISSN =	{1868-8969},
  year =	{2025},
  volume =	{326},
  editor =	{Endrullis, J\"{o}rg and Schmitz, Sylvain},
  publisher =	{Schloss Dagstuhl -- Leibniz-Zentrum f{\"u}r Informatik},
  address =	{Dagstuhl, Germany},
  URL =		{https://drops.dagstuhl.de/entities/document/10.4230/LIPIcs.CSL.2025.36},
  URN =		{urn:nbn:de:0030-drops-227933},
  doi =		{10.4230/LIPIcs.CSL.2025.36}
}

@article{kozen1997kat,
author = {Kozen, Dexter},
title = {Kleene algebra with tests},
year = {1997},
issue_date = {May 1997},
publisher = {Association for Computing Machinery},
address = {New York, NY, USA},
volume = {19},
number = {3},
issn = {0164-0925},
url = {https://doi.org/10.1145/256167.256195},
doi = {10.1145/256167.256195},
journal = {ACM Trans. Program. Lang. Syst.},
month = may,
pages = {427–443},
numpages = {17}
}

@inproceedings{conway1971regular,
  title={Regular algebra and finite machines},
  author={John H. Conway},
  year={1971},
  url={https://api.semanticscholar.org/CorpusID:117721644}
}

@InProceedings{Kozen96KAT,
author="Kozen, Dexter",
editor="Margaria, Tiziana
and Steffen, Bernhard",
title="Kleene algebra with tests and commutativity conditions",
booktitle="Tools and Algorithms for the Construction and Analysis of Systems",
year="1996",
publisher="Springer Berlin Heidelberg",
address="Berlin, Heidelberg",
pages="14--33",
isbn="978-3-540-49874-2"
}

@article{Ibarra78,
author = {Ibarra, Oscar H.},
title = {Reversal-Bounded Multicounter Machines and Their Decision Problems},
year = {1978},
issue_date = {Jan. 1978},
publisher = {Association for Computing Machinery},
address = {New York, NY, USA},
volume = {25},
number = {1},
issn = {0004-5411},
url = {https://doi.org/10.1145/322047.322058},
doi = {10.1145/322047.322058},
journal = {J. ACM},
month = jan,
pages = {116–133},
numpages = {18}
}


\newpage
\appendix
\section{Appendix: Full Proof}
\subsection{Proof for Section \ref{sec:background}}
\begin{theorem}[Coincidence of $CKA$ and $CKA^{*}$]
For all expressions $e_1,e_2\in T_{\Sigma}$,
\[
CKA \vDash e_1 = e_2
\quad\Longleftrightarrow\quad
CKA^{*} \vDash e_1 = e_2 .
\]
\end{theorem}

\begin{proof}
The direction $CKA \Rightarrow CKA^{*}$ is immediate.

For the converse, recall that in $CKA$,
\[
CKA \vdash e_1 = e_2
\;\Longleftrightarrow\;
P(L(e_1))=P(L(e_2)),
\]
so it suffices to show that the additional \(*\)-continuity axiom of $CKA^{*}$
preserves Parikh images.  The only extra axiom is
\[
a\, b^{*}\, c
=\sum_{n\ge0} a\, b^{n}\, c.
\]

Since $P(xy)=P(x)\oplus P(y)$ with $S_1 \oplus S_2 = \{\,u+v \mid u\in S_1,\; v\in S_2 \,\}$, and $P(b^{*})=\bigcup_{n\ge0}P(b^{n})$, we have
\[
P(a\, b^{*}\, c)
  = P(a) \oplus P(b^{*}) \oplus P(c)
  = P(a) \oplus \Bigl(\bigcup_{n\ge0} P(b^{n})\Bigr) \oplus P(c)
\]
\[
\phantom{P(a\, b^{*}\, c)}
  = \bigcup_{n\ge0} \bigl( P(a)\oplus P(b^{n})\oplus P(c) \bigr)
  = P\!\left(\sum_{n\ge0} a\, b^{n}\, c\right).
\]
Thus \(*\)-continuity preserves Parikh equivalence, hence any $CKA^{*}$ proof
reduces to a $CKA$ proof.  
\end{proof}

\subsection{Proof for Section \ref{sec:coin}}

\begin{lemma}[Commutativity of Star]\label{lem:comm}
For all expressions $p,q$, if $\ka{KA}{C}{pq}{qp}$, then
\[
\ka{KA}{C}{p^{*} q}{q p^{*}}.
\]
\end{lemma}
\begin{proof}
Assume $\ka{KA}{C}{pq}{qp}$, i.e.\ $pq = qp$ is derivable in $KA+C$.

We first show that
\[
q + p^{*} q p \;\le\; p^{*} q.
\tag{$\star$}
\]

Indeed:
\[
\begin{aligned}
q + p^{*} q p
  &= q + p^{*} (q p) \\
  &= q + p^{*} (p q) && \text{(since $pq = qp$)} \\
  &= (1 + p^{*} p)\, q \\
  &= p^{*} q .
\end{aligned}
\]
This proves $(\star)$.

Now apply the KA \emph{right-induction} rule:
\[
\text{if } b + c a \le c \text{ then } b a^{*} \le c.
\]
Instantiate
\[
b := q,\qquad
a := p,\qquad
c := p^{*} q .
\]
Since $(\star)$ is precisely $q + (p^{*} q) p \le p^{*} q$, the rule yields
\[
q p^{*} \;\le\; p^{*} q.
\]

Vice versa, $ p^{*} q \;\le\;q p^{*}$.

Thus
\[
q p^{*} = p^{*} q,
\]
establishing the claim.
\end{proof}

\begin{corollary}[Commutativity of $b^{\le n}$]\label{lem:local-com}
For all KA expressions~$e$, we have $\ka{KA}{C}{b^{\le n}e} {e b^{\le n}}$.
\end{corollary}\begin{proof}
    We proceed in two steps. First we show that $b$ itself commutes with every
    expression $e$ over the alphabet $\{a,b,c\}$. Then we lift this property
    to the finite segment $b^{\le n}$ by induction on $n$.
    
    \medskip\noindent
    \emph{Step~1: $b$ commutes with every expression.}
    We claim that for all expressions $e$,
    \[
      \ka{KA}{C}{b e}{e b}.
    \]
    We prove this by structural induction on $e$.
    
    \smallskip\noindent
    \emph{Base cases.}
    If $e = 0$ or $e = 1$, then
    \[
      b0 = 0 = 0b,
      \qquad
      b1 = b = 1b,
    \]
    so $\ka{KA}{C}{b e}{e b}$ holds.
    
    If $e$ is one of the primitive letters $a,b,c$, it directly holds since $(a,b), (b,c) \in I$.
    
    \smallskip\noindent
    \emph{Sum.}
    If $e = e_1 + e_2$, then using distributivity and the induction hypothesis
    for $e_1$ and $e_2$ we obtain
    \[
      b(e_1 + e_2)
      = b e_1 + b e_2
      = e_1 b + e_2 b
      = (e_1 + e_2) b.
    \]
    
    \smallskip\noindent
    \emph{Product.}
    If $e = e_1 e_2$, then by associativity of multiplication and the induction
    hypothesis for $e_1$ and $e_2$,
    \[
      be = b e_1 e_2
      = (b e_1) e_2
      = (e_1 b) e_2
      = e_1 (b e_2)
      = e_1 (e_2 b)
      = e b.
    \]
    
    \smallskip\noindent
    \emph{Star.}
    If $e = p^{*}$ for some expression $p$, then by the induction hypothesis
    we have $\ka{KA}{C}{b p}{p b}$.
    Instantiating Lemma~\ref{lem:local-com} (Commutativity of Star) with
    this pair $(p,q) = (p,b)$ yields
    \[
      \ka{KA}{C}{p^{*} b}{b p^{*}}.
    \]
    By symmetry of equality, this is equivalent to
    $\ka{KA}{C}{b p^{*}}{p^{*} b}$.
    
    This completes the structural induction and proves that
    \[
      \ka{KA}{C}{b e}{e b}
    \]
    holds for all expressions $e$ over $\{a,b,c\}$.
    
    \medskip\noindent
    \emph{Step~2: $b^{\le n}$ commutes with every expression.}
    Fix $e$ and prove by induction on $n \in \mathbb{N}$ that
    \[
      \ka{KA}{C}{b^{\le n} e}{e\, b^{\le n}}.
    \]
    
    For $n = 0$ we have $b^{\le 0} = 1$, hence
    \[
      b^{\le 0} e = 1 \cdot e = e = e \cdot 1 = e b^{\le 0}.
    \]
    
    Assume now that $\ka{KA}{C}{b^{\le n} e}{e b^{\le n}}$ holds for some $n$.
    By definition of $b^{\le n+1}$ and distributivity of multiplication over
    addition,
    \[
      b^{\le n+1} e
      = (b^{\le n} + b^{n+1}) e
      = b^{\le n} e \;+\; b^{n+1} e.
    \]
    By the induction hypothesis, the first summand rewrites to $e b^{\le n}$.
    
    From Step~1, we know that $b$ commutes with $e$, i.e.\ $b e = e b$.
    By associativity of multiplication, a straightforward induction on $k$
    shows that $b^k e = e b^k$ holds for all $k \in \mathbb{N}$.
    In particular, $\ka{KA}{C}{b^{n+1} e}{e b^{n+1}}$.
    Thus
    \[
      b^{\le n+1} e
      = e b^{\le n} + e b^{n+1}
      = e (b^{\le n} + b^{n+1})
      = e b^{\le n+1},
    \]
    where the last equality is again the definition of $b^{\le n+1}$ and
    distributivity.
    
    This completes the induction on $n$ and proves that
    $\ka{KA}{C}{b^{\le n} e}{e b^{\le n}}$ holds for all $n \in \mathbb{N}$
    and all expressions $e$.
\end{proof}

\begin{theorem}
Let $\Sigma_i$ be one
equivalence class under~$C$.
For any expressions $p,q \in T_{\Sigma_i}$,
\[
  \ka{KA}{\{ab = ba : a,b \in \Sigma_i\}}{pq}{qp}.
\]
\end{theorem}

\begin{proof}
We argue by structural induction on the pair $(p,q)$ of expressions over
$\Sigma_i$.  Let $C(p,q)$ denote
\[
KA + \{ab = ba \mid a,b\in\Sigma_i\} \;\vdash\; pq = qp.
\]

\begin{itemize}
  \item \textbf{Both generators.}
    If $p = a$ and $q = b$ with $a,b \in \Sigma_i$, then by assumption
    $ab = ba$, so $pq = qp$ and $C(p,q)$ holds.

  \item \textbf{Zero and one.}
    If $p = 0$ or $p = 1$, then $pq = 0 = qp$ or $pq = q = qp$
    by the KA axioms.  The symmetric cases $q = 0$ or $q = 1$
    are analogous.

  \item \textbf{Addition in $p$ (or $q$).}  
    Suppose $p = p_1 + p_2$.
    By the induction hypothesis,
    \[
      p_1 q = q p_1
      \quad\text{and}\quad
      p_2 q = q p_2.
    \]
    Then, using distributivity,
    \[
      (p_1 + p_2) q
        = p_1 q + p_2 q
        = q p_1 + q p_2
        = q (p_1 + p_2),
    \]
    so $C(p,q)$ holds.  The case $q = q_1 + q_2$ is symmetric.

  \item \textbf{Multiplication in $p$ (or $q$).}
    Suppose $p = p_1 p_2$.
    By the induction hypothesis,
    \[
      p_1 q = q p_1,
      \quad
      p_2 q = q p_2.
    \]
    Then
    \[
    (p_1 p_2) q
      = p_1 (p_2 q)
      = p_1 (q p_2)
      = (p_1 q) p_2
      = (q p_1) p_2
      = q (p_1 p_2),
    \]
    so $C(p,q)$ holds.  The case $q = q_1 q_2$ is symmetric.

  \item \textbf{Star in $p$ (or $q$).}
    By induction hypothesis and theorem \ref{lem:local-com}, we directly prove the case.
\end{itemize}

Thus we have $pq = qp$ for all such $p,q\in Reg~\Sigma_i$, as required.
\end{proof}

\begin{lemma}
Every word \(w \in \Sigma^{*}\) admits exactly one factorization under \(C\).
\end{lemma}
\begin{proof}
If \(w=1\), then the only factorization under \(C\) is the single block \(1\).

Now suppose \(w\neq 1\). Since the equivalence classes of \(C\) partition the alphabet into subalphabets \(\Sigma_1,\dots,\Sigma_m\), we obtain a factorization of \(w\) by cutting \(w\) into maximal consecutive segments whose letters all lie in the same subalphabet. This yields a factorization under \(C\).

Uniqueness is immediate from maximality: once the word is partitioned into maximal consecutive segments coming from a single subalphabet, there is no freedom to merge two adjacent blocks, since they come from different subalphabets, and no freedom to split a block, since it is already maximal inside one subalphabet. Hence the factorization is unique.
\end{proof}

\begin{theorem}
For every expression \(p\) and its factorized form \(\hat{p}\), we have
\[
L(p)=L(\hat{p}).
\]
\end{theorem}

\begin{proof}
Write the matrix representation of \(p\) as
\[
p = u^{T} A^{*} v,
\qquad
A = \sum_{a \in \Sigma} a \cdot A_a,
\]
and the factorized expression as
\[
\hat{p}
= \bigl(u^{(m)}\bigr)^{T}\, (\hat{A})^{*}\, v^{(m)}.
\]

We prove both inclusions.

\paragraph{(1) \(L(p)\subseteq L(\hat{p})\).}
Let \(w \in L(p)\).  
By Theorem~\ref{thm:uniq-fac}, let its unique factorization under \(I\) be
\[
w = w_1 w_2 \cdots w_j,
\qquad
w_i \in \Sigma_{k_i}^{*},\ k_i \neq k_{i+1}.
\]

\smallskip
\emph{Case \(w = 1\).}  
Since every nonzero entry of \(A\) is a word of positive length, the only way
to obtain the empty word from \(u^{T}A^{*}v\) is through the constant path
\(u^{T}v\).  
Thus \(1 = u^{T}v = (u^{(m)})^{T} v^{(m)}\), hence \(1 \in L(\hat{p})\).

\smallskip
\emph{Case \(w \neq 1\).}  
Since \(A = \sum_{i=1}^{m} A_i\), we have
\[
A^{*}
  = \bigl(\sum_{i=1}^{m} A_i\bigr)^{*},
\qquad
A_i = \sum_{a\in\Sigma_i} a\cdot A_a.
\]
Hence
\[
w \in L\bigl(u^{T} A_{k_1}^{|w_1|} A_{k_2}^{|w_2|} \cdots A_{k_j}^{|w_j|} v\bigr)
  \subseteq
L\bigl(u^{T} A_{k_1}^{+} A_{k_2}^{+} \cdots A_{k_j}^{+} v\bigr).
\]

By construction of \(\hat{A}\), any block path of the form
\[
A_{k_1}^{+} A_{k_2}^{+} \cdots A_{k_j}^{+}
\]
appears as a summand in the \((k_1,k_j)\)-th $n\times n$ matrix entry of the $nm\times nm$ matrix \((\hat{A})^j\); let this
entry be the matrix block \(B\).
Thus \(A_{k_1}^{+} A_{k_2}^{+} \cdots A_{k_j}^{+} \le B\), and therefore
\[
L\bigl(u^{T} A_{k_1}^{+} A_{k_2}^{+} \cdots A_{k_j}^{+} v\bigr)
  \subseteq
L\bigl(u^{T} B v\bigr)
  \subseteq
L\bigl((u^{(m)})^{T} (\hat{A})^{j} v^{(m)}\bigr).
\]

Since \((\hat{A})^{n} \le (\hat{A})^{*}\), it follows that
\[
L\bigl((u^{(m)})^{T} (\hat{A})^{j} v^{(m)}\bigr)
   \subseteq
L\bigl((u^{(m)})^{T} (\hat{A})^{*} v^{(m)}\bigr)
   = L(\hat{p}),
\]
establishing the desired inclusion.

\paragraph{(2) \(L(\hat{p})\subseteq L(p)\).}
Since \(A = \sum_{i=1}^{m}A_i\), we have \(A_i^{+} \le A^{*}\).  
Define the block matrix
\[
B =
\begin{pmatrix}
A^{*} & A^{*} & \cdots & A^{*} \\
A^{*} & A^{*} & \cdots & A^{*} \\
\vdots & \vdots & \ddots & \vdots \\
A^{*} & A^{*} & \cdots & A^{*}
\end{pmatrix}.
\]
Clearly \(\hat{A} \le B\) and \(1 \le B\).

Moreover, since \(A^{*}A^{*}=A^{*}\), we have \(BB = B\).
By the KA induction axiom (if \(ab\le b\) then \(a^{*}b\le b\)), this implies
\(B^{*}B \le B\).  
Together with \(1\le B\), we obtain
\[
B^{*} \le 1 + B^{*}B \le B.
\]
Since \(B \le B^{*}\), it follows that
\[
B = B^{*}.
\]

Thus
\[
L(\hat{p})
 \subseteq L\bigl( (u^{(m)})^{T} B^{*} v^{(m)} \bigr)
 = L\bigl( (u^{(m)})^{T} B v^{(m)} \bigr)
 = L(u^{T}A^{*} v)
 = L(p).
\]

Both inclusions hold, so \(L(p)=L(\hat{p})\).
\end{proof}

\begin{theorem}
For every semilinear set \(S\subseteq\mathbb{N}^k\), there exists a regular
expression \(e\) such that \(P(L(e))=S\).
\end{theorem}

\begin{proof}
Let \(a_i\) denote the letter whose Parikh image contributes to the \(i\)-th
coordinate.  
For any vector \(v=(v_1,\ldots,v_k)\in\mathbb{N}^k\), it is immediate that the
expression
\[
  e_v \;=\; \prod_{i=1}^k a_i^{\,v_i}
\]
has Parikh image \(P(e_v)=v\).

Now let the semilinear set be written in standard form:
\[
  S
  = \bigcup_{i=1}^n
    \left\{
      \mathbf{b}_i + \sum_{j=1}^{m_i} n_j \mathbf{p}_{ij}
      \;\middle|\;
      n_j\in\mathbb{N}
    \right\}.
\]
For each base vector \(\mathbf{b}_i\), choose an word \(w_{\mathbf{b}_i}\)
with \(P(w_{\mathbf{b}_i})=\mathbf{b}_i\); similarly choose word
\(w_{\mathbf{p}_{ij}}\) with \(P(w_{\mathbf{p}_{ij}})=\mathbf{p}_{ij}\).
Then
\[
  e
  = \sum_{i=1}^{n}
      \Bigl(
        w_{\mathbf{b}_i}
        \;\prod_{j=1}^{m_i} w_{\mathbf{p}_{ij}}^{\,*}
      \Bigr)
\]
satisfies \(P(L(e))=S\).
\end{proof}

\begin{theorem}[Commutation Equivalence via Factorization]
Let $w_1, w_2 \in \Sigma^*$ with $C$-factorizations
\[
w_1 = w_{11} w_{12} \cdots w_{1n_1}
\quad\text{and}\quad
w_2 = w_{21} w_{22} \cdots w_{2n_2}.
\]
Then
\[
w_1 \equiv_C w_2
\quad\Longleftrightarrow\quad
n_1 = n_2
\;\text{ and }\;
\forall i \in [1,n_1],\;
\Psi(w_{1i}) = \Psi(w_{2i}).
\]
\end{theorem}

\begin{proof}
($\Rightarrow$)  
If $w_1 \equiv_C w_2$, then commutation steps can only permute letters
within equivalence classes, not across them.  
Hence their $I$-factorizations have the same number of blocks $n_1 = n_2$,
and each corresponding block belongs to the same class $\Sigma_i$.
By Theorem~\ref{thm:cka-completeness}, within each class we have
$\Psi(w_{1i}) = \Psi(w_{2i})$.

(\(\Leftarrow\))
If \(n_1=n_2\) and each corresponding pair of factors satisfies
\(\Psi(w_{1i})=\Psi(w_{2i})\), then, since each block has finite length and all letters within a block are mutually commutative, a finite sequence of commutations transforms \(w_{1i}\) into \(w_{2i}\) for every \(i\). Hence
we have $w_{1i} \equiv_C w_{2i}$ for all $i$.
we conclude $w_1 \equiv_C w_2$.
\end{proof}
\subsection{Proof for Section \ref{sec:undec}}

\subsubsection{Proof for Section \ref{subsec:halting}}
\paragraph{Full analysis of invalid execution}
We begin by analyzing the possible forms of a legal transition of $M$. There are four cases:
We first distinguish the possible forms of a legal transition. The first three cases describe transitions in which the tape head is scanning a non-blank symbol, and the last case describes transitions in which the tape head is scanning the blank symbol.

\begin{itemize}
    \item \textbf{The tape head scans a non-blank symbol and does not move.}
    Suppose
    \[
    ID_1=xqsyB,
    \]
    where $x,y\in\{0,1\}^*$ and $s\in\{0,1\}$. In this case, a single transition can change only the current state $q$ and the scanned symbol $s$, producing a configuration of the form
    \[
    ID_2=xq's'yB.
    \]

    \item \textbf{The tape head scans a non-blank symbol and moves left.}
    Suppose
    \[
    ID_1=xs_1qs_2yB,
    \]
    where $x,y\in\{0,1\}^*$ and $s_1,s_2\in\{0,1\}$. In this case, a left move can change only the current state $q$ and the scanned symbol $s_2$, producing a configuration of the form
    \[
    ID_2=xq's_1s_2'yB.
    \]

    \item \textbf{The tape head scans a non-blank symbol and moves right.}
    Suppose
    \[
    ID_1=xqsyB,
    \]
    where $x,y\in\{0,1\}^*$ and $s\in\{0,1\}$. In this case, a right move can change only the current state $q$ and the scanned symbol $s$, producing a configuration of the form
    \[
    ID_2=xs'q'yB.
    \]

    \item \textbf{The tape head scans the blank symbol.}
    Suppose
    \[
    ID_1=xqB.
    \]
    Then the tape head is scanning a blank symbol. In this case, the machine may stay in place, move left, or move right, producing one of the following configurations:
    \[
    ID_2=xq'B,\qquad ID_2=xq'sB,\qquad ID_2=xs'q'B.
    \]
\end{itemize}
From these cases, we see that a single transition can affect at most three consecutive symbols of a configuration. Therefore, if one finds a mismatch between the corresponding $i$th, $(i+1)$st, and $(i+2)$nd symbols of $\#ID_1$ and $\#ID_2$, then the transition must be illegal. It therefore suffices to consider the following cases.

\begin{itemize}
    \item \textbf{Case (1): one of the configurations has length $2$.}

    Suppose either $|ID_1|=2$ or $|ID_2|=2$. Since every configuration must contain at least one state symbol $q\in Q$ and one blank symbol $B$, the minimum possible configuration length is $2$. Moreover, by our assumption that the machine never writes a blank symbol, a single transition can either preserve the length of a configuration or increase it by $1$. Therefore, the only possible legal transitions involving a configuration of length $2$ are those satisfying
    \[
    |ID_1|=|ID_2|=2
    \qquad\text{or}\qquad
    |ID_2|=|ID_1|+1=3.
    \]
    Since there are only finitely many such pairs, we can explicitly enumerate all illegal words of the form
    \[
    \#ID_1\#ID_2 b^{|\#ID_1\#ID_2|}
    \]
    with $|ID_1|<3$ and $|ID_2|<3$. We denote the corresponding finite sum by
    \[
    \sum_{\substack{\#ID_1\#ID_2\text{ is illegal}\\ |ID_1|<3,\ |ID_2|<3}}
    \#ID_1\#ID_2 b^{|\#ID_1\#ID_2|}.
    \]

    For the remaining cases, where one of the configurations has length $2$ and the other has length strictly greater than $3$, we use the expression
    \[
    \#b(\Sigma_{ID}b)^2\#b(\Sigma_{ID}b)^{>3}
    \;+\;
    \#b(\Sigma_{ID}b)^{>3}\#b(\Sigma_{ID}b)^2
    \]
    to cover all such illegal pairs.

    We write $e_{\mathit{illegal}<3}(M)$ for the union of the two expressions above.

    \item \textbf{Case (2): no state symbol and no blank symbol appear in the local window.}

    Suppose that among the six symbols consisting of the $i$th, $(i+1)$st, and $(i+2)$nd positions of both $\#ID_1$ and $\#ID_2$, there is neither a state symbol $q\in Q$ nor a blank symbol $B$. Then these positions are away from both the head and the boundary blank, so the tape contents of $(i+1)$st, and $(i+2)$nd positions must be copied unchanged. Hence the corresponding symbols must agree. Any mismatch is therefore illegal. We use a finite set of $6$-tuples, denoted by $\Delta_1$, to collect all such illegal local patterns.

    \item \textbf{Case (3): a state symbol appears in the local window, but no blank symbol does.}

    Suppose that one of the six local symbols contains a state symbol $q\in Q$, but none of them is $B$. Then the window lies near the tape head, but not near the right boundary blank. In this case, legality can be checked directly from the transition function $\delta$. We use another finite set of $6$-tuples, denoted by $\Delta_2$, to collect all such illegal local patterns.

    \item \textbf{Case (4): a blank symbol appears in the local window.}

    Suppose that a blank symbol $B$ appears among these six local symbols. If $B$ does not appear in the relevant portion of $ID_1$ but does appear in $ID_2$, then the transition is illegal, since by assumption the machine never writes a blank symbol and hence cannot shorten the written portion of the tape. Therefore, if $B$ appears, it must already occur at the $(i+2)$nd position of $ID_1$.

    Moreover, depending on whether the transition preserves the length of the configuration or increases it by one, the blank symbol in $ID_2$ must appear either at the $(i+2)$nd or the $(i+3)$rd position in order for the transition to be legal. Accordingly, we use a finite set $\Delta_3$ to collect all local $6$-tuples such that the transition is illegal regardless of the $(i+3)$rd symbol, and another finite set $\Delta_4$ to collect all local patterns such that the transition is illegal whenever the $(i+3)$rd symbol is not $B$.
\end{itemize}

As one can see, every possible error in a transition between two consecutive configurations is covered by one of the four cases above. If one of the configurations has length $2$, then the pair is handled by Case~(1). If either
\[
|ID_2|>|ID_1|+1
\qquad\text{or}\qquad
|ID_2|<|ID_1|,
\]
then the violation is detected by Case~(4). Any copying error away from the tape head is detected by Case~(2), while any error near the tape head, including an incorrect local transition, is detected by Case~(3).

We take $\Delta=\Delta_1\cup\Delta_2\cup\Delta_3$ and $\Delta'=\Delta_4$.
\newline

\begin{theorem}[Halting Checking]
Let
\[
M=(Q,\Gamma,\delta,q_0,q_f)
\]
be a Turing machine on a semi-infinite tape that never overwrites any symbol with the blank symbol and has a capturing state \(q_c\). Let \(\overline{H}(M)\) be the expression constructed above from \(M\). Then
\[
\ka{KA^*}{C}{\overline{H}(M)}{\Sigma^*}
\qquad\text{if and only if}\qquad
M \text{ does not halt}.
\]
\end{theorem}

\begin{proof}
Since equivalence in \(KA^*+C\) coincides with equality of the corresponding regular languages, it suffices to reason at the level of regular languages.

We prove both directions.

\medskip
\noindent
(\(\Rightarrow\)) Suppose
\[
\ka{KA^*}{C}{\overline{H}(M)}{\Sigma^*}.
\]
We show that \(M\) does not halt. Suppose, towards a contradiction, that \(M\) halts. Let
\[
\#ID_1\#\cdots\#ID_n
\]
be the encoding of a halting computation of \(M\). Define
\[
w=\#ID_1\#\cdots\#ID_n\, b^{|\#ID_1\#\cdots\#ID_n|}.
\]
We claim that
\[
w\notin L_C(\overline{H}(M)).
\]

Indeed, since the number of occurrences of \(b\) in \(w\) is exactly equal to the number of all other symbols, while the non-\(b\) part of \(w\) is exactly a syntactically valid halting trace, we have
\[
w\notin L_C(e_{\mathit{illegal}}).
\]
Furthermore, every consecutive pair \(ID_i,ID_{i+1}\) in this trace forms a legal transition of the machine, and hence
\[
w\notin L_C(e_{\mathit{invalid\_execution}}(M)).
\]
Finally, since this computation halts rather than entering the capturing state, the symbol \(q_c\) does not occur in \(w\). Therefore,
\[
w\notin L_C(\Sigma^*q_c\Sigma^*).
\]

Combining the above facts, we obtain
\[
w\notin L_C(\overline{H}(M)).
\]
On the other hand, clearly
\[
w\in L_C(\Sigma^*).
\]
This contradicts the assumption
\[
\ka{KA^*}{C}{\overline{H}(M)}{\Sigma^*}.
\]
Therefore, \(M\) does not halt.

\medskip
\noindent
(\(\Leftarrow\)) Suppose \(M\) does not halt. Thus \(M\) has no halting execution. By construction, every word \(w\notin L_C(\overline{H}(M))\) would have to encode a halting execution of \(M\): it would not belong to \(L_C(e_{\mathit{illegal}})\), so it would be a syntactically valid trace with the correct number of \(b\)'s; it would not belong to \(L_C(e_{\mathit{invalid\_execution}}(M))\), so every consecutive pair of configurations would form a legal transition; and it would not belong to \(L_C(\Sigma^*q_c\Sigma^*)\), so the execution would not enter the capturing state. Hence \(w\) would encode a genuine halting execution of \(M\), contradicting the assumption that \(M\) does not halt. Therefore,
\[
L_C(\overline{H}(M))=L_C(\Sigma^*),
\]
and so
\[
\ka{KA^*}{C}{\overline{H}(M)}{\Sigma^*}.
\]
This completes the proof.
\end{proof}
\subsubsection{Proof for Section \ref{subsec:c-looping}}
\begin{theorem}[Legal and Illegal Traces]
Let
\[
\Sigma=\{0,1,B,\#\}\cup Q\cup\{b\},
\qquad
C=\{(\sigma,b)\mid \sigma\in \{0,1,B,\#\}\cup Q\}.
\]
Then
\[
\ka{KA}{C}{e_{\mathit{traces}}+e_{\mathit{illegal}}}{\Sigma^*}.
\]
\end{theorem}

\begin{proof}
We begin the proof via a claim:
\[
\ka{KA}{C}{e_{\mathit{match\_b}}+e_{\mathit{mismatch\_b}}}{\Sigma^*}.
\]

We prove the claim using the coincidence theorem, Theorem~\ref{thm:coin}.

Consider first another alphabet \(\Sigma'=\{a',b'\}\) with commutativity condition
\[
C'=\{(a',b')\}.
\]
Define
\[
e'_{\mathit{mismatch\_b}}=b'^+(a'b')^*+a'(a'(1+b'))^*
\qquad\text{and}\qquad
e'_{\mathit{match\_b}}=(a'b')^*.
\]
Then \(e'_{\mathit{match\_b}}\) denotes the set of words in which the numbers of \(a'\) and \(b'\) are equal, while \(e'_{\mathit{mismatch\_b}}\) denotes the set of words in which these numbers differ. Hence
\[
L_{C'}(e'_{\mathit{mismatch\_b}}+e'_{\mathit{match\_b}})
=
L_{C'}(\Sigma'^*).
\]

Since regular-language equivalence coincides with \( * \)-continuous equivalence by Theorem~\ref{thm:equiv}, and \( * \)-continuous equivalence coincides with general equivalence whenever \(C'\) is transitive by Theorem~\ref{thm:coin}, it follows that
\[
\ka{KA}{\{a'b'=b'a'\}}{e'_{\mathit{mismatch\_b}}+e'_{\mathit{match\_b}}}{\Sigma'^*}.
\]

Now substitute \(a'\) by \(\Sigma_{\setminus b}\) and \(b'\) by \(b\). We obtain
\[
\ka{KA}{\{\Sigma_{\setminus b}b=b\Sigma_{\setminus b}\}}
{e_{\mathit{mismatch\_b}}+e_{\mathit{match\_b}}}{\Sigma^*}.
\]

Finally, from the commutativity conditions in \(C\), we have
\[
\ka{KA}{C}{\Sigma_{\setminus b}b}{b\Sigma_{\setminus b}}.
\]
Therefore,
\[
\ka{KA}{C}{e_{\mathit{match\_b}}+e_{\mathit{mismatch\_b}}}{\Sigma^*}.
\]

Next we will show
\[
\ka{KA}{C}{e_{\mathit{traces}}+e_{\mathit{match\_b}} \setminus e_{\mathit{traces}}}{e_{\mathit{match\_b}}}.
\]

By construction,
\[
L(e_{\mathit{traces}}+e_{\mathit{match\_b}} \setminus e_{\mathit{traces}})
=
L(e_{\mathit{match\_b}}).
\]
Hence, by completeness of $KA$ for regular languages, we have
\[
\kar{e_{\mathit{traces}}+e_{\mathit{match\_b}} \setminus e_{\mathit{traces}}}{e_{\mathit{match\_b}}}.
\]
It follows immediately that
\[
\ka{KA}{C}{e_{\mathit{traces}}+e_{\mathit{match\_b}} \setminus e_{\mathit{traces}}}{e_{\mathit{match\_b}}}.
\]
Thus 
\[
\ka{KA}{C}{e_{\mathit{traces}}+e_{\mathit{illegal}}}{\Sigma^*}.
\]

\end{proof}

Next, we are going to prove a useful intermediate lemma showing that any consecutive execution that increases the length of configuration by length 2 is not valid.

\begin{theorem}[Invalid Length Increase]\label{thm:invalid-length}
Let $n\in\mathbb{N}$. Define
\[
e_{ID}
=
\#b(0b+1b)^*\Bigl(\sum_{q\in Q} qb\Bigr)(0b+1b)^*Bb
\]
to represent a single valid configuration. Next, define
\[
e_{|ID|\le n}
=
\sum_{n_1+n_2=n}
\Bigl(
\#b(0b+1b)^{\le n_1}
\Bigl(\sum_{q\in Q} qb\Bigr)
(0b+1b)^{\le n_2}Bb
\Bigr)
\]
and
\[
e_{|ID|\ge n}
=
\sum_{n_1+n_2=n}
\Bigl(
\#b(0b+1b)^{\ge n_1}
\Bigl(\sum_{q\in Q} qb\Bigr)
(0b+1b)^{\ge n_2}Bb
\Bigr).
\]
where $e^{\geq n}=e^ne^*$ and $e^{\leq n}=\sum\limits_{i\le n}e^i$.
These expressions denote, respectively, the sets of configurations whose written tape content has length at most $n$ and at least $n$.

Then for any restricted Turing machine $M$ described above
\[
\kale{KA}{C}{e_{|ID|\le n}\,e_{|ID|\ge n+2}}{e_{\mathit{invalid\_execution}}(M)}.
\]
\end{theorem}

\begin{proof}
We distinguish two cases.

\medskip
\noindent
\textbf{Case 1: $n=0$.}
In this case,
\[
e_{|ID|\le 0}
=
\#b\Bigl(\sum_{q\in Q} qb\Bigr)Bb,
\]
and
\[
e_{|ID|\ge 2}
=
\sum_{n_1+n_2=2}
\Bigl(
\#b(0b+1b)^{\ge n_1}
\Bigl(\sum_{q\in Q} qb\Bigr)
(0b+1b)^{\ge n_2}Bb
\Bigr).
\]
Here, $e_{|ID|\le 0}$ represents the shortest possible configurations. On the other hand, it is easy to see that
\[
\kale{KA}{C}{e_{|ID|\ge 2}}{\#b(\Sigma_{ID}b)^{>3}}.
\]
Therefore,
\[
\kale{KA}{C}{e_{|ID|\le 0}\,e_{|ID|\ge 2}\,\#b}{e_{\mathit{illegal}<3}(M)},
\]
and hence
\[
\kale{KA}{C}{e_{|ID|\le 0}\,e_{|ID|\ge 2}\,\#b}{e_{\mathit{invalid\_execution}}(M)}.
\]

\medskip
\noindent
\textbf{Case 2: $n\ge 1$.}
In this case, we can always select three letters from each of the configurations represented by $e_{|ID|\le n}$ and $e_{|ID|\ge n+2}$. Let $s$ be the position of the blank symbol $B$ in a word from $e_{|ID|\le n}$. By construction, the corresponding positions $s$ and $s+1$ in a word from $e_{|ID|\ge n+2}$ cannot contain $B$.

Indeed, every configuration in $e_{|ID|\ge n+2}$ has written tape content of length at least $n+2$, so its total length is strictly greater than $n+4$ once the additional symbols $\#$ and $B$ are taken into account. Hence we may unroll the first $n+3$ symbols of such a configuration, which corresponds to the first $2n+6$ symbols of the expression, since each symbol is paired with a $b$. Thus
\[
L(e_{|ID|\ge n+2})
=
L\Bigl(
\sum_{|w|=2n+6} w e_w
\Bigr),
\]
where each $e_w$ is the corresponding residual expression (derivative of $e_{|ID|\geq n+2}$ w.r.t $w$). By completeness of $KA$, we obtain
\[
\kar{e_{|ID|\ge n+2}}{\sum_{|w|=2n+6} w e_w},
\]
and therefore
\[
\ka{KA}{C}{e_{|ID|\ge n+2}}{\sum_{|w|=2n+6} w e_w}.
\]

Similarly, we may unroll all words in $e_{|ID|\le n}$ and obtain
\[
\ka{KA}{C}{e_{|ID|\le n}}{\sum_{w \in L(e_{|ID|})} w}.
\]

Now let $w_1$ be any word arising from $e_{|ID|\le n}$, and let $w_2$ be any prefix of length $2n+6$ arising from $e_{|ID|\ge n+2}$. By construction, the symbol $B$ occurs at position $|w_1|-1$ in $w_1$, whereas in $w_2$ the first $n+2$ symbols cannot contain $B$. Therefore, this discrepancy is always detected by one of the patterns in $\Delta_3$ or $\Delta_4$. Hence
\[
\kale{KA}{C}{w_1w_2}{
\Biggl(
\sum_{(l_1,l_2,l_3,l_4,l_5,l_6)\in \Delta_3\cup\Delta_4}
\Bigl(
l_1bl_2bl_3b(\Sigma_{ID}b)^*\#(\Sigma_{ID}b)^*l_4bl_5bl_6b
\Bigr)
\Biggr)
\Sigma_{ID\setminus B}(\Sigma_{ID}b)^*
}.
\]

For each residual expression $e_w$, we trivially have
\[
\kale{KA}{C}{e_w}{(\Sigma_{ID}b)^*}.
\]
Combining the above inclusions, we conclude that
\[
\kale{KA}{C}{e_{|ID|\le n}\,e_{|ID|\ge n+2}}{e_{\mathit{invalid\_execution}}(M)}.
\]
\end{proof}

\begin{theorem}[Case Analysis]
Let
\[
M=(Q,\Gamma,\delta,q_0,q_f)
\]
be a restricted Turing machine as described above. Write
\[
e_{\mathit{init}}=\#bq_0bBb,
\qquad
e_{ID}=\#b(0b+1b)^*\Bigl(\sum_{q\in Q}qb\Bigr)(0b+1b)^*Bb.
\]
Then for every \(n\in\mathbb{N}\),
\[
KA + C \vDash
e_{\mathit{init}}e_{ID}^{n}
\le
e_{\mathit{invalid\_execution}}(M)
+
\sum_{\substack{|w|\le \frac{n^2+7n+6}{2}\\ w\in L(e_{\mathit{init}}e_{ID}^{n})}} w.
\]
\end{theorem}

\begin{proof}
By Theorem~\ref{thm:invalid-length}, any trace containing two consecutive configurations whose lengths differ by at least \(2\) is accepted by \(e_{\mathit{invalid\_execution}}(M)\). In particular,
\[
\kale{KA}{C}{e_{|ID|\le n}\,e_{|ID|\ge n+2}}{e_{\mathit{invalid\_execution}}(M)}.
\]
It therefore remains to consider only those traces in which the size of the configuration increases by at most \(1\) at each step.

Since every configuration has length at least \(3\), accounting for \(\#\), a state symbol \(q\), and the blank symbol \(B\), every such trace of \(n\) steps has total length at most
\[
\sum_{k=0}^{n}(k+3)=\frac{n^2+7n+6}{2}.
\]
Hence, at the level of regular languages, we have
\[
L(e_{\mathit{init}}e_{ID}^{n})
\subseteq
L\Biggl(
\sum_{0\le i\le n}
e_{ID}^i\,e_{|ID|\le i}\,e_{|ID|\ge i+2}\,e_{ID}^{\,n-i-1}
\Biggr)
\;\cup\;
L\Biggl(
\sum_{\substack{|w|\le \frac{n^2+7n+6}{2}\\ w\in L(e_{\mathit{init}}e_{ID}^{n})}} w
\Biggr).
\]

By completeness of \(KA\), it follows that
\[
\kale{KA}{C}{e_{\mathit{init}}e_{ID}^{n}}
{\sum_{0\le i\le n}
e_{ID}^i\,e_{|ID|\le i}\,e_{|ID|\ge i+2}\,e_{ID}^{\,n-i-1}
+
\sum_{\substack{|w|\le \frac{n^2+7n+6}{2}\\ w\in L(e_{\mathit{init}}e_{ID}^{n})}} w }.
\]
Combining this with Theorem~\ref{thm:invalid-length}, we conclude that
\[
KA + C \vDash
e_{\mathit{init}}e_{ID}^{n}
\le
e_{\mathit{invalid\_execution}}(M)
+
\sum_{\substack{|w|\le \frac{n^2+7n+6}{2}\\ w\in L(e_{\mathit{init}}e_{ID}^{n})}} w.
\]
\end{proof}

\begin{theorem}[$c$-loop checking]
Let
\[
M=(Q,\Gamma,\delta,q_0,q_f)
\]
be a Turing machine on a semi-infinite tape, which never overwrites any symbol by the blank symbol and has a capturing state $q_c$. Let $\overline{H}(M)$ be the expression constructed above from $M$. If $M$ is $c$-looping, then
\[
\ka{KA}{C}{\overline{H}(M)}{\Sigma^*}.
\]
\end{theorem}

\begin{proof}
It suffices to show
\[
\ka{KA}{C}{\overline{H}(M)}{\Sigma^*}
\qquad\text{and}\qquad
\ka{KA}{C}{\Sigma^*}{\overline{H}(M)}.
\]

The first inequality is immediate: since
\[
\karle{\overline{H}(M)}{\Sigma^*},
\]
we also have
\[
\kale{KA}{C}{\overline{H}(M)}{\Sigma^*}.
\]

For the converse direction, by Theorem~\ref{thm:trace}, we have
\[
\ka{KA}{C}{\Sigma^*}{e_{\mathit{illegal}}+e_{\mathit{traces}}}.
\]
Since
\[
\overline{H}(M)
=
e_{\mathit{illegal}}
+
e_{\mathit{invalid\_execution}}(M)
+
\Sigma^*q_c\Sigma^*,
\]
it remains to show that
\[
\kale{KA}{C}{e_{\mathit{traces}}}{e_{\mathit{invalid\_execution}}(M)+\Sigma^*q_c\Sigma^*}.
\]

Recall that
\[
\begin{aligned}
e_{\mathit{traces}}
={}&\#bq_0bBb\\
&\bigl(\#b(0b+1b)^*(\sum_{q\in Q}qb)(0b+1b)^*Bb\bigr)^*\\
&\#b(0b+1b)^*q_fb(0b+1b)^*Bb.
\end{aligned}
\]
For convenience, write
\[
e_{\mathit{init}}=\#bq_0bBb,
\qquad
e_{ID}=\#b(0b+1b)^*\Bigl(\sum_{q\in Q}qb\Bigr)(0b+1b)^*Bb,
\qquad
e_{\mathit{end}}=\#b(0b+1b)^*q_fb(0b+1b)^*Bb.
\]
Then
\[
e_{\mathit{traces}}=e_{\mathit{init}}\,e_{ID}^*\,e_{\mathit{end}}.
\]

Since \(M\) is \(c\)-looping, there exists \(n\in\mathbb{N}\) such that after \(n\) steps every valid execution has entered the capturing state \(q_c\). Hence, if we inspect the first \(n\) configurations of a trace, then either the trace already contains the capturing state \(q_c\), or an execution error occurs within the first \(n\) steps.

We therefore separate all traces according to whether they finish within \(n+1\) steps or require more than \(n+1\) steps:
\[
\ka{KA}{C}{e_{\mathit{init}}e_{ID}^*e_{\mathit{end}}}
{\sum_{0\le i\le n} e_{\mathit{init}}e_{ID}^{i}e_{\mathit{end}}
+
e_{\mathit{init}}e_{ID}^{n+1}e_{ID}^*e_{\mathit{end}}}.
\]

We first show that the traces finishing within \(n+1\) steps are all accepted by \(e_{\mathit{invalid\_execution}}(M)\). By Theorem~\ref{thm:case-analysis}, for each \(0\le i\le n\),
\[
KA + C \vDash
e_{\mathit{init}}e_{ID}^{i+1}
\le
e_{\mathit{invalid\_execution}}(M)
+
\sum_{\substack{|w|\le \frac{(i+1)^2+7(i+1)+6}{2}\\ w\in L(e_{\mathit{init}}e_{ID}^{i+1})}} w.
\]
Since
\[
KA + C \vDash
e_{\mathit{init}}e_{ID}^{i}e_{\mathit{end}}
\le
e_{\mathit{init}}e_{ID}^{i+1},
\]
it follows that
\[
KA + C \vDash
e_{\mathit{init}}e_{ID}^{i}e_{\mathit{end}}
\le
e_{\mathit{invalid\_execution}}(M)
+
\sum_{\substack{|w|\le \frac{(i+1)^2+7(i+1)+6}{2}\\ w\in L(e_{\mathit{init}}e_{ID}^{i+1})}} w.
\]
Now, since \(M\) is \(c\)-looping, every finite execution finishing within at most \(n\) steps must always be invalid. Hence, by finite case analysis,
\[
KA + C \vDash
\sum_{\substack{|w|\le \frac{(i+1)^2+7(i+1)+6}{2}\\ w\in L(e_{\mathit{init}}e_{ID}^{i+1})}} w
\le
e_{\mathit{invalid\_execution}}(M).
\]
Therefore,
\[
KA + C \vDash
e_{\mathit{init}}e_{ID}^{i}e_{\mathit{end}}
\le
e_{\mathit{invalid\_execution}}(M)
\qquad
\text{for every }0\le i\le n.
\]

It remains to consider the execution traces of more than \(n+1\) steps:
\[
e_{\mathit{init}}e_{ID}^{n+1}e_{ID}^*e_{\mathit{end}}.
\]
Again by Theorem~\ref{thm:case-analysis}, we have
\[
KA + C \vDash
e_{\mathit{init}}e_{ID}^{n+1}
\le
e_{\mathit{invalid\_execution}}(M)
+
\sum_{\substack{|w|\le \frac{(n+1)^2+7(n+1)+6}{2}\\ w\in L(e_{\mathit{init}}e_{ID}^{n+1})}} w.
\]
Since \(M\) is \(c\)-looping after \(n\) steps, each such finite word
\[
w\in L(e_{\mathit{init}}e_{ID}^{n+1})
\]
either contains the state symbol \(q_c\), or else contains two consecutive configurations \(ID_i,ID_{i+1}\) that do not form a valid transition. Hence, for every such \(w\), finite case analysis gives
\[
\kale{KA}{C}{w}{e_{\mathit{invalid\_execution}}(M)+\Sigma^*q_c\Sigma^*}.
\]
Therefore,
\[
\kale{KA}{C}{w\,e_{ID}^*e_{\mathit{end}}}{e_{\mathit{invalid\_execution}}(M)+\Sigma^*q_c\Sigma^*}.
\]
Combining these two facts, we obtain
\[
\kale{KA}{C}{e_{\mathit{init}}e_{ID}^{n+1}e_{ID}^*e_{\mathit{end}}}
{e_{\mathit{invalid\_execution}}(M)+\Sigma^*q_c\Sigma^*}.
\]

Putting everything together, we conclude that
\[
\kale{KA}{C}{e_{\mathit{traces}}}{e_{\mathit{invalid\_execution}}(M)+\Sigma^*q_c\Sigma^*},
\]
and hence
\[
\ka{KA}{C}{\Sigma^*}{\overline{H}(M)}.
\]
This completes the proof.
\end{proof}

\subsubsection{Proof for Section \ref{subsec:undec}}

\begin{lemma}[Recursive inseparability of returning $0$ and returning $1$]\label{thm:inseparable-01}
Let $\Phi$ be a computable one-to-one encoding function from Turing
machines to natural numbers. Define
\[
A_0=\{\,M \mid M \text{ halts on input } \Phi(M) \text{ and returns }0\,\}
\]
and
\[
A_1=\{\,M \mid M \text{ halts on input } \Phi(M) \text{ and returns }1\,\}.
\]
Then $A_0$ and $A_1$ are recursively inseparable: there is no decidable
set $C$ of Turing machines such that
\[
A_0\subseteq C
\qquad\text{and}\qquad
A_1\cap C=\emptyset .
\]
\end{lemma}

\begin{proof}
First observe that $A_0$ and $A_1$ are disjoint, since a deterministic
Turing machine cannot return both $0$ and $1$ on the same input.

Suppose, for contradiction, that there exists a decidable set $C$ of
Turing machines such that
\[
A_0\subseteq C
\qquad\text{and}\qquad
A_1\cap C=\emptyset .
\]
Since $C$ is decidable, there is a Turing machine $D$ which, on input
$n$, decides whether the machine encoded by $n$ belongs to $C$. We now
construct a Turing machine $N$ as follows. On input $n$, the machine
$N$ first uses $D$ to decide whether the machine encoded by $n$ is in
$C$. If it is in $C$, then $N$ halts and returns $1$; otherwise, $N$
halts and returns $0$.

Equivalently,
\[
N(n)=
\begin{cases}
1, & \text{if the machine encoded by } n \text{ is in } C,\\
0, & \text{otherwise.}
\end{cases}
\]
In particular, consider the behavior of $N$ on its own code
$\Phi(N)$. There are two cases.

If $N\in C$, then by the definition of $N$ we have
\[
N(\Phi(N))=1.
\]
Thus $N\in A_1$. But this contradicts the assumption that
$A_1\cap C=\emptyset$, since we are in the case $N\in C$.

On the other hand, if $N\notin C$, then by the definition of $N$ we have
\[
N(\Phi(N))=0.
\]
Thus $N\in A_0$. Since $A_0\subseteq C$, it follows that $N\in C$,
again a contradiction.

Both cases lead to contradictions. Therefore no such decidable separator
$C$ exists, and hence $A_0$ and $A_1$ are recursively inseparable.
\end{proof}

\begin{lemma}[Recursive inseparability of $c$-looping and halting]
For the class of Turing machines \(M\) that operate on a semi-infinite tape, have a capturing state \(q_c\), and never overwrite a tape symbol with a blank symbol, the sets
\[
\{\,M \mid M \text{ is \(c\)-looping}\,\}
\qquad\text{and}\qquad
\{\,M \mid M \text{ halts}\,\}
\]
are recursively inseparable.
\end{lemma}

\begin{proof}
For every standard Turing machine \(M\) and input \(w\), we can effectively construct a Turing machine \(M'\) in the restricted class above such that if \(M\) halts on input \(w\) and returns \(0\), then \(M'\) is \(c\)-looping, while if \(M\) halts on input \(w\) and returns \(1\), then \(M'\) halts.

Therefore, if there were a decidable set separating the \(c\)-looping machines from the halting machines in this restricted class, then, by applying the effective transformation \((M,w)\mapsto M'\), we would obtain a decidable set separating the machines that halt and return \(0\) from those that halt and return \(1\). This contradicts the recursive inseparability result from Lemma~\ref{thm:inseparable-01}.
\end{proof}

\begin{lemma}[Minimal Encoding]
Let
\[
M=(Q,\Gamma,\delta,q_0,q_f)
\]
be a Turing machine on a semi-infinite tape with a capturing state \(q_c\), and suppose that \(M\) never overwrites any symbol with a blank symbol. Then there exists an effective translation \(T\) such that \(T(M)\) is a regular expression over an alphabet \(\Sigma\) containing at least the three letters \(\{a,b,c\}\), equipped with a commutativity condition \(C\) such that
\[
(a,b)\in C,\qquad (b,c)\in C,\qquad (a,c)\notin C.
\]
The following hold:
\begin{itemize}
    \item if \(M\) is \(c\)-looping, then
    \[
    \ka{KA}{C}{T(M)}{\Sigma^*};
    \]
    \item if
    \[
    \ka{KA^*}{C}{T(M)}{\Sigma^*},
    \]
    then \(M\) does not halt.
\end{itemize}
\end{lemma}

\begin{proof}
In the previous construction, we already defined a translation from $M$ to a regular expression over the alphabet
\[
\Sigma'=\{0,1,B,\#\}\cup Q\cup\{b\},
\]
with commutativity condition
\[
C'=\{(\sigma,b)\mid \sigma\in \{0,1,B,\#\}\cup Q\}.
\]
We now reduce this alphabet to $\{a,b,c\}$.

Keep the symbol $b$ unchanged, and encode each symbol in
\[
\{0,1,B,\#\}\cup Q
\]
by a distinct binary word over $\{a,c\}$. Since the set
\[
\{0,1,B,\#\}\cup Q
\]
is finite, such an injective encoding exists. Let
\[
f:\Sigma'\to \{a,b,c\}^*
\]
denote this encoding, where $f(b)=b$ and $f(\sigma)\in\{a,c\}^*$ for every $\sigma\neq b$. We choose $f$ so that its image is unambiguous, that is, every encoded word admits a unique decomposition into codewords.

Now define
\[
e_{\mathit{complement}}
=
\Sigma^*\setminus \Bigl(\sum_{\sigma\in\Sigma'} f(\sigma)\Bigr)^*.
\]
Then, by construction,
\[
\kar{e_{\mathit{complement}}+\Bigl(\sum_{\sigma\in\Sigma'} f(\sigma)\Bigr)^*}{\Sigma^*}.
\]

We define
\[
T(M)=e_{\mathit{complement}}+f(\overline{H}(M)),
\]
where $f(\overline{H}(M))$ is obtained from $\overline{H}(M)$ by replacing each symbol $\sigma\in\Sigma'$ by its codeword $f(\sigma)$.

Suppose first that $M$ is $c$-looping. By Theorem~\ref{thm:c-loop-raw}, we have
\[
\ka{KA}{C'}{\overline{H}(M)}{\Bigl(\sum_{\sigma\in\Sigma'} \sigma\Bigr)^*}.
\]
Applying the encoding $f$, we obtain
\[
\ka{KA}{\{(f(\sigma),b)\mid \sigma\in \{0,1,B,\#\}\cup Q\}}
{f(\overline{H}(M))}
{\Bigl(\sum_{\sigma\in\Sigma'} f(\sigma)\Bigr)^*}.
\]
Since every codeword $f(\sigma)$ lies in $\{a,c\}^*$, the commutativity relations above are derivable from
\[
\{(a,b),(b,c)\}\subseteq C.
\]
Hence
\[
\ka{KA}{C}{f(\overline{H}(M))}{\Bigl(\sum_{\sigma\in\Sigma'} f(\sigma)\Bigr)^*}.
\]
Combining this with
\[
\kar{e_{\mathit{complement}}+\Bigl(\sum_{\sigma\in\Sigma'} f(\sigma)\Bigr)^*}{\Sigma^*},
\]
we conclude that
\[
\ka{KA}{C}{T(M)}{\Sigma^*}.
\]

For the converse direction, suppose that
\[
\ka{KA^*}{C}{T(M)}{\Sigma^*}.
\]
We show that $M$ does not halt. Assume, for contradiction, that $M$ halts. Let
\[
w=\#ID_1\#\cdots\#ID_n\, b^{|\#ID_1\#\cdots\#ID_n|}
\]
be the encoded halting trace used in Theorem~\ref{thm:halting-raw}. Since the encoding $f$ is unambiguous, the word $f(w)$ belongs to $\Bigl(\sum_{\sigma\in\Sigma'} f(\sigma)\Bigr)^*$ and cannot belong to $e_{\mathit{complement}}$. Moreover, exactly the same argument as in Theorem~\ref{thm:halting-raw} shows that
\[
f(w)\notin L_C\bigl(f(\overline{H}(M))\bigr).
\]
Therefore,
\[
f(w)\notin L_C(T(M)).
\]
On the other hand, clearly
\[
f(w)\in L_C(\Sigma^*).
\]
This contradicts the assumption
\[
\ka{KA^*}{C}{T(M)}{\Sigma^*}.
\]
Hence $M$ does not halt.
\end{proof}

\end{document}